\documentclass{article}

\usepackage[english]{babel}

\usepackage[letterpaper,top=2cm,bottom=2cm,left=3cm,right=3cm,marginparwidth=1.75cm]{geometry}

\usepackage{amsmath}
\usepackage{enumerate}
\usepackage{graphicx}
\usepackage[colorlinks=true, allcolors=blue]{hyperref}
\usepackage{amsthm,amssymb}
\usepackage{mathtools}
\usepackage{comment}
\usepackage{mathrsfs}

\theoremstyle{plain}
\newtheorem{theorem}{Theorem}[section]
\newtheorem{corollary}{Corollary}[theorem]
\newtheorem{proposition}[theorem]{Proposition}
\newtheorem{lemma}[theorem]{Lemma}

\theoremstyle{definition}
\newtheorem{definition}{Definition}[section]

\newtheorem{problem}[definition]{Problem}

\newtheorem{assumption}{Assumption}[section]
\newtheorem*{assumption*}{Assumption}

\theoremstyle{remark}
\newtheorem{remark}{Remark}[section]

\everymath{\displaystyle}

\DeclareMathOperator*{\diag}{diag}

\title{Convergence Rates of Turnpike Theorems for Portfolio Choice in Stochastic Factor Models}
\author{Hiroki Yamamichi\footnote{
Graduate School of Engineering Science, The University of Osaka, 1-3, Machikaneyama, Toyonaka, Osaka, 560-8531, Japan, Email: \texttt{yamamichi@sigmath.es.osaka-u.ac.jp}
}}
\date{}

\begin{document}
\maketitle

\begin{abstract}
Turnpike theorems state that if an investor's utility is asymptotically equivalent to a power utility, then the optimal investment strategy converges to the CRRA strategy as the investment horizon tends to infinity. This paper aims to derive the convergence rates of the turnpike theorem for optimal feedback functions in stochastic factor models. In these models, optimal feedback functions can be decomposed into two terms: myopic portfolios and excess hedging demands. We obtain convergence rates for myopic portfolios in nonlinear stochastic factor models and for excess hedging demands in quadratic term structure models, where the interest rate is a quadratic function of a multivariate Ornstein--Uhlenbeck process. We show that the convergence rates are determined by (i) the decay speed of the price of a zero-coupon bond and (ii) how quickly the investor's utility becomes power-like at high levels of wealth. As an application, we consider optimal collective investment problems and show that sharing rules for terminal wealth affect convergence rates.
\end{abstract}

\section{Introduction}
Since the seminal works of Merton \cite{M69, M71}, optimal investment problems for continuous-time models have been developed in various directions. In particular, stochastic factor models have been used to capture \textit{stochastic investment opportunity sets}, such as the predictability of stock returns, stochastic volatility, and stochastic interest rates. For an overview of optimal investment problems with stochastic factor models, we refer the reader to the review paper \cite{Z}.
In these models, optimal investment strategies can be decomposed into two terms, namely, \textit{myopic portfolios} and \textit{excess hedging demands}. This means that at first investors choose the myopic portfolios as if the investment opportunity sets are constant, then they adjust their portfolios by adding the excess hedging demands to adapt to future changes in the investment environment. Although computing these terms typically requires analyzing fully nonlinear Hamilton--Jacobi--Bellman (HJB) equations, the computations become tractable for homothetic utilities, such as exponential, power, and log utilities. As a result, explicit formulas for optimal investment strategies are typically available for these special utilities. In contrast, it is challenging to analyze optimal investment strategies for generic utilities, and, to the best of our knowledge, only a few works have calculated the optimal strategies, including Detemple and Rindisbacher \cite{DGR}, Fukaya \cite{F}, Lakner \cite{L}, Ocone and Karatzas \cite{OK}, and Putschögl and Sass \cite{PS}.

Turnpike theorems fill this gap between power utilities and general utilities. Informally, turnpike theorems state that if a utility function is similar to a power utility at large wealth levels, then both the optimal wealth process and the optimal investment strategy converge to those for a power utility as the investment horizon tends to infinity. This paper aims to derive the convergence rates of turnpike theorems for both myopic portfolios and excess hedging demands in stochastic factor models. 

Turnpike theorems originate in the classic work of von~Neumann \cite{Ne} in economic growth theory. In the context of optimal portfolios, Mossin \cite{M} first proves portfolio turnpikes in discrete-time settings under the assumption that a utility function $U$ has affine risk tolerance, that is, $- \frac{U^\prime(x)}{U^{\prime\prime}(x)} = a x + b$.  Mossin's results are extended by Leland \cite{L}, Ross \cite{R}, and Hakansson \cite{H} to include general utility functions. Huberman and Ross \cite{HR} derive a necessary and sufficient condition for the turnpike property. Cox and Huang \cite{CH} prove the first turnpike theorem in continuous-time settings using martingale methods under the assumption that there exist constants $A_1, A_2, b, z^\ast > 0$ such that
\begin{equation}\label{cox.asp}
        \left|(U^\prime)^{-1}(z) - A_1 z^{- \frac{1}{b}}  \right| \leq A_2 z^{-a}, \quad z \in (0, z^\ast] 
\end{equation}
holds for some $ a \in \left[0, 1 / b \right)$. Jin \cite{J} extends their results to include consumption. 
Huang and Zariphopoulou \cite{HZ} show that the condition 
\[
    \lim_{x \nearrow \infty} \frac{U^\prime(x)}{x^{\gamma - 1}} = K
\]
for some $K > 0$, which is weaker than the condition (\ref{cox.asp}), is sufficient for the turnpike property, using viscosity solutions for the associated HJB equations. Dybvig et~al.\ \cite{DRB} prove portfolio turnpikes for complete markets in the Brownian filtration without assuming stationary investment opportunity sets. Guasoni et~al.\ \cite{GKRX} consider general incomplete market models that include one-dimensional stochastic factor models, and they prove three types of turnpikes (abstract, classic, and explicit).  

Although it is important to know how fast turnpike theorems hold in practice, the above works do not derive convergence rates. However, Bian and Zheng \cite{BZ} first estimate the convergence rate of the turnpike property under the Black--Scholes model. Beyond the classical expected utility framework, Geng and Zariphopoulou \cite{GZ} recently studied turnpike-type limiting properties and the convergence rate for the forward relative risk tolerance function under time-monotone forward performance criteria in Ito diffusion markets. Since Geng and Zariphopoulou \cite{GZ} focus on the subclass of time-monotone forward utilities, there is still no result on convergence rates for excess hedging demands. In the present paper, we aim to derive convergence rates for both myopic portfolios and excess hedging demands in stochastic factor models.

Our contributions are fourfold. First, we derive convergence rates of turnpike theorems for myopic portfolios with general stochastic factor models in complete markets (Theorem~\ref{myopic.turnpike}), which extends the results of Bian and Zheng \cite{BZ} from the Black--Scholes model to more general settings. In particular, we find that the convergence rate is determined by the price of a zero-coupon bond and the rate at which the investor's utility becomes power-like at high levels of wealth. As made precise in Remark~\ref{exp.decay.ZC}, the rate is typically exponential. Moreover, we prove uniform convergence in the wealth variable for portfolio proportions (Theorem~\ref{uni.conv.myopic}), which has not yet been documented. Furthermore, we also derive the convergence rate for the optimal portfolio processes (Theorem~\ref{wealth.turnpike}). 

Second, we also derive convergence rates of turnpike theorems for excess hedging demands with quadratic term structure models (Theorem~\ref{excess.turnpike}), where the instantaneous rate is a quadratic function of the stochastic factor process. 
To the best of our knowledge, no previous studies have derived the convergence rates for excess hedging demands. We find that the convergence rates are the same as those of myopic portfolios and uniform convergence in wealth holds (Theorem~\ref{uni.excess.turnpike}).

Third, by applying our main results, we analyze the turnpike properties of the optimal strategies for optimal collective investment problems. In these problems, there are $n$ investors who delegate their portfolio management to a fund manager. The fund manager invests on their behalf to optimize the expected social utility, constructed from the individual utilities and a sharing rule, according to which the fund manager allocates the terminal wealth to individuals. We consider two sharing rules: a Pareto-optimal sharing rule and a linear sharing rule. We find that these sharing rules affect convergence rates (Theorems \ref{Pareto.tunpike} and \ref{Linshre.turnpike}). In particular, we show that the convergence rate for a linear sharing rule is faster than for a Pareto-optimal sharing rule when the least risk-averse investor among $n$ investors is no less risk-averse than the log investor.

Finally, we methodologically provide a probabilistic approach based on martingale duality methods and Malliavin calculus, in contrast to the PDE techniques used in prior work by Bian and Zheng \cite{BZ}. When applying Malliavin calculus to optimal investment problems, previous research papers such as \cite{L, PS} often assume conditions that depend on the investment horizon. In the present paper, we assume conditions (Assumptions \ref{coeff.asp}--\ref{smooth.asp}) that are independent of the investment horizon, allowing us to apply Malliavin calculus techniques when the investment horizon tends to infinity.

The rest of this paper is organized as follows. The main results are discussed in Sect.~\ref{main.result}, which consists of four subsections. In Sect.~\ref{flow.rep.sec}, we provide stochastic flow representations of optimal feedback strategies for general utilities. Sect.~\ref{Myopic.Turnpike} estimates the convergence rate of the turnpike theorem for myopic portfolios with general factor models. Sect.~\ref{Excess.Turnpike} estimates the rate for excess hedging demands with quadratic term structure models. In Sect.~\ref{Application.sec}, we offer applications of our main results to optimal collective investment problems. Sect.~\ref{proof.section} contains proofs of our main results. The appendix contains short reviews of Malliavin calculus, option pricing theory in stochastic factor models, the relationship between stochastic control methods and martingale duality methods, and the matrix Riccati equation. 

\section{Main results}\label{main.result}

\subsection{Stochastic flow representation of optimal feedback functions}\label{flow.rep.sec}
Let $(\Omega, \mathcal{F}, \mathbb{P}, (\mathcal{F}_t)_{t \in[0, T]})$ be a filtered probability space endowed with $(\mathcal{F}_t)_{t \in [0, T]}$, the augmentation of the natural filtration generated by the $n$-dimensional Brownian motion $W \coloneqq (W^1, \dots, W^n)^\top$. We consider a financial market with one riskless bond and $n$ risky assets. The price processes of the riskless bond $S^0$ and $n$ risky assets $S = (S^1, \dots, S^n)^\top$ are modeled as
\begin{equation}\label{gen.stoch.fac.model}
    \begin{aligned}
        dS^0_t &= S^0_t r(Y_t) dt,& S^0_0 &= 1, \\
        dS_t &= \diag(S_t) \{ \mu (Y_t) dt + \sigma (Y_t) dW_t\}, & S_0 &= s_0 \in \mathbb{R}^n_{++}, \\
        dY_t &= b(Y_t) dt + a (Y_t) dW_t,& Y_0 &= y \in \mathbb{R}^m,
    \end{aligned}
\end{equation}
where $r:\mathbb{R}^m \to \mathbb{R}, \mu:\mathbb{R}^m \to \mathbb{R}^n, \sigma: \mathbb{R}^m \to \mathbb{R}^{n \times n},  b:\mathbb{R}^m \to \mathbb{R}^m,  a:\mathbb{R}^m \to \mathbb{R}^{m \times n}$ satisfy Assumption~\ref{coeff.asp} below, $\mathrm{diag}(x)$ denotes the diagonal $n \times n$ matrix whose $(i, i)$-th element is component $x_i$ of $x = (x_1, \dots, x_n)^\top \in \mathbb{R}^n$, and the $m$-dimensional process $Y = (Y^1, \dots, Y^m)^\top$ is referred to as the stochastic factor processes, which affect the coefficients of the asset price processes. We denote the market price of risk as $\theta(y) \coloneqq \sigma(y)^{-1} \{ \mu(y) - r(y) \mathbf{1} \}$, where $\mathbf{1} = (1, \dots, 1)^\top \in \mathbb{R}^n$. Furthermore, we let $\beta = (\beta_t)_{t \in [0, T]}$ be a discounted price process, $Z = (Z_t)_{t \in [0, T]}$ be a likelihood ratio process, and $H = (H_t)_{t \in [0, T]}$ be a state price density process as follows:
\begin{align*}
    \beta_t &\coloneqq \frac{1}{S^0_t} = \exp \left( - \int_0^t r(Y_s) ds \right), \\
    Z_t &\coloneqq \exp \left\{ - \int_0^t \theta^\top (Y_s) dW_s - \frac{1}{2} \int_0^t |\theta (Y_s)|^2 ds \right\},\\
    H_t &\coloneqq \beta_t Z_t.
\end{align*}
We assume some conditions on these coefficient functions.
\begin{assumption}\label{coeff.asp}
\begin{enumerate}[(i)]
    \item All functions $r, \mu, \sigma, b, a, \theta$ on $\mathbb{R}^m$ are continuous functions.
    \item $\sigma(y)$ is invertible for all $y \in \mathbb{R}^m$.
    \item $\theta$ has linear growth.
    \item $b, a$ are Lipschitz continuous and $a$ is bounded.
\end{enumerate}
\end{assumption}Under Assumption~\ref{coeff.asp}, a local martingale $Z$ is a martingale. For the proof, see \cite[Section~6.2]{LS}. Therefore, we can define a unique equivalent martingale measure $\mathbb{Q}$ and $n$-dimensional $\mathbb{Q}$-Brownian motion $W^\mathbb{Q}$ as follows: 
\begin{align*}
    d\mathbb{Q} |_{\mathcal{F}_t} &\coloneqq Z_t d\mathbb{P}|_{\mathcal{F}_t} , \\
    W^\mathbb{\mathbb{Q}}_t &\coloneqq W_t + \int_0^t \theta(Y_s) ds.
\end{align*}

$\pi = (\pi_t)_{t \in [0, T]}$ is a portfolio process if $\pi$ is an $n$-dimensional $(\mathcal{F}_t)_t$-progressively measurable process and satisfies 
\[
    \int_0^T |\pi_t|^2 dt < \infty \quad P-a.s.
\]
For each initial wealth $x \geq 0$ and portfolio process $\pi$, we define the corresponding wealth process $X^{x, \pi} = (X^{x, \pi}_t)_{t \in [0, T]} $ as the solution to the following SDE:
\[
    dX^{x, \pi}_t = \left[r(Y_t) X^{x, \pi}_t + \pi^\top_t(\mu(Y_t) - r(Y_t)\mathbf{1}) \right] dt + \pi^\top_t \sigma(Y_t) dW_t, \quad X_0 = x.
\]
Given $x \geq 0$, we say that a portfolio process $\pi$ is admissible at $x$ if the corresponding wealth process $X^{x, \pi}$ satisfies
\[
    X^{x, \pi}_t \geq 0, \quad t \in [0, T] 
\]
almost surely.

An investor's risk preference is represented by a utility function $U$.
\begin{definition}\label{Util.def}
    We call $U:(0, \infty) \to \mathbb{R}$ a utility function if $U$ is strictly increasing, strictly concave, and twice continuously differentiable on $(0, \infty)$ and satisfies the Inada conditions
    \[
        \lim_{x \searrow 0} U^\prime (x) = \infty, \quad \lim_{x \nearrow \infty} U^\prime(x) = 0.
    \]
\end{definition}
Let $I \coloneqq (U^\prime)^{-1} :(0, \infty) \to (0, \infty)$ be the inverse marginal utility $U^\prime$. By the definition of $U$, $I$ is strictly decreasing and continuously differentiable on $(0, \infty)$ and satisfies
\[
    \lim_{z \searrow 0} I(z) = \infty, \quad \lim_{z \nearrow \infty} I(z) = 0.
\]

The investor in this paper desires to maximize their expected utility and find an optimal portfolio process $\hat{\pi}$.
This problem is formulated as follows.
\begin{problem}
    Find an optimal $\pi \in \mathcal{A}(x)$ for the problem 
    \[
        \sup_{\pi \in \mathcal{A}(x)} \mathbb{E} \left[U(X^{x, \pi}_T) \right]
    \]
    of maximizing expected utility from terminal wealth, where
    \[
        \mathcal{A}(x) \coloneqq \left\{ \pi; \text{$\pi$ is admissible at $x$},  \mathbb{E} \left[U(X^{x, \pi}_T)^{-} \right] < \infty \right\}.
    \]
\end{problem}

To use the martingale method, we assume the following growth conditions.
\begin{assumption}\label{martingale.method.asp}
    \begin{itemize}
        \item [(i)] There exist $r_0 \in \mathbb{R},\; r_1 \in \mathbb{R}^m$ such that for every $y \in \mathbb{R}^m$,
        \[
            r(y) \geq r_0 + r_1^\top y.
        \]
        \item [(ii)] There exist $\kappa > 0, \; \rho \in (0, 1]$ such that for every $z > 0$, 
        \[
            I(z) \leq \kappa (1 + z^{-\rho}).
        \]
    \end{itemize}
\end{assumption}

Under Assumptions \ref{coeff.asp} and \ref{martingale.method.asp}, we can use the martingale method as in \cite[Theorem~3.7.6]{KS}. For the proof, see Sect.~\ref{flow.rep.proof}.
\begin{theorem}\label{Optimal Terminal Wealth}
    Under Assumptions \ref{coeff.asp} and \ref{martingale.method.asp}, for each $x > 0$, there exists $\hat{\lambda} > 0$ such that $x = \mathbb{E}[H_T I(\hat{\lambda}H_T)]$.
    The optimal terminal wealth $\xi$ and the optimal wealth process $\hat{X} = (\hat{X}_t)_{t \in [0, T]}$ are 
    \[
        \xi = I(\hat{\lambda} H_T), \quad \hat{X}_t = \frac{1}{H_t} \mathbb{E}_t \left[H_T I(\hat{\lambda} H_T) \right].
    \]
    Moreover, the optimal portfolio process $\hat{\pi} = (\hat{\pi}_t)_{t \in [0, T]}$ is given by 
    \[
        \hat{\pi}_t = (\sigma^\top (Y_t))^{-1} \left( \frac{\psi_t}{H_t} + X_t \theta(Y_t) \right),
    \]
    where $\psi$ is the integrand in the stochastic integral representation $M_t = x + \int_0^t (\psi_u)^\top dW_u$ of the martingale $\left( \mathbb{E}_t [H_T I(\hat{\lambda}H_T)]  \right)_{t \in [0, T]}$. 
\end{theorem}

To derive explicit formulas for $\psi$, we assume regularity and growth conditions for the market coefficients, $r, \theta, b, a$, and the derivative of $I$, which enable us to apply the results of the Malliavin calculus.
\begin{assumption}\label{smooth.asp}
    \begin{enumerate}
        \item[(i)] $r, \theta$ are continuously differentiable and of polynomial growth, and their Jacobian matrices $Dr, D\theta$ are also of polynomial growth.
        \item[(ii)] $b, a$ are continuously differentiable, and the Jacobian matrices $Db, Da_{\cdot l}:\mathbb{R}^m \to \mathbb{R}^{m \times m}$ are bounded, where $Db = \left( \frac{\partial b_i}{\partial x_j} \right)_{\substack{1 \leq i \leq m \\ 1 \leq j \leq m}}, \; Da_{\cdot l} = \left( \frac{a_{i,l}}{\partial x_j} \right)_{\substack{ 1 \leq i \leq m \\ 1 \leq j \leq m}},\; (l = 1, \dots, n)$. 
        \item [(iii)] There exists $K > 0$ such that for any $y, z_1, \dots, z_n \in \mathbb{R}^m$,
        \[
            y^\top \left( b(y) - a(y) \theta(y) \right) + \sum_{j = 1}^n z_j^\top \left( D b(y) - \sum_{l = 1}^n D a_{\cdot l}(y) \theta^l(y) \right)z_j \leq K \left( 1 + |y|^2 + \sum_{j = 1}^n |z_j|^2 \right).
        \]
        \item [(iv)] $|zI^\prime(z)| \leq \kappa (1 + z^{-\rho}) $ for some $\kappa > 0, \; \rho \in (0, 1]$.
    \end{enumerate}
\end{assumption}

For Theorem~\ref{explicit.portfolio}, we prepare some notation. Let a pair $(Y,Z)$ of $\mathbb{R}^m$-valued stochastic process $Y$ and $\mathbb{R}^{m \times m}$-valued stochastic process $Z$ be the solution to the following system of SDEs:
\begin{align*}
    dY_s &= b(Y_s) ds + a(Y_s) dW_s,& Y_t &= y, \\
    dZ_s &= Db(Y_s) Z_s ds + \sum_{j = 1}^n D a_{\cdot j}(Y_s) Z_s dW^j_s,& Z_t &= I. 
\end{align*}
Then a pair $(Y, Z)$ is a Markov process and $(Y^{(t, y)}, Z^{(t, y)})$ denotes the solution to the above system of SDEs when $(Y, Z)$ starts from $(y, I) \in \mathbb{R}^{m} \times \mathbb{R}^{m \times m}$ at time 0. Note that $Z^{(t, y)}$ always starts from the identity matrix $I \in \mathbb{R}^{m \times m}$ and $Y^{(0, y)}_s = Y^{(t, Y_t^{(0, y)})}_s$ for $s \in [t, T]$. Because $Z^{(t, y)}$ can be thought of as the derivative of $Y^{(t, y)}$ with respect to the initial value $y$, we use the notation $\nabla_y Y^{(t, y)} \coloneqq Z^{(t, y)}$ instead of $Z^{(t, y)}$.
Furthermore, let $H^{(t, y)} = \left( H_s^{(t, y)} \right)_{s \in [t, T]} $ be a state price density process that starts at time t and is given by
\[
    H_s^{(t, y)} \coloneqq  \exp \left( - \int_t^s r(Y^{(t, y)}_u) du - \int_t^s \theta^\top (Y^{(t, y)}_u) dW_u - \frac{1}{2} \int_t^s |\theta (Y^{(t, y)}_u)|^2 du \right).
\]
Let $\nabla_y H^{(t, y)} = \left( \nabla_y H_s^{(t, y)} \right)_{s \in [t, T]}, \; L^{(t, y)} = \left( L_s^{(t, y)}\right)_{s \in [t, T]} $ be $\mathbb{R}^m$-valued stochastic processes given by 
\begin{align*}
    L_s^{(t, y)} &\coloneqq \int_t^s (Dr(Y^{(t, y)}_u) \nabla_y Y^{(t, y)}_u )^\top du + \int_t^s \left( D\theta (Y^{(t, y)}_u) \nabla_y Y^{(t, y)}_u \right)^\top  dW_u \\
    & \qquad + \int_t^s \left( D\theta (Y^{(t, y)}_u) \nabla_y Y_u^{(t, y)} \right)^\top \theta(Y_u^{(t, y)}) du, \\
     \nabla_y H_s^{(t, y)} &\coloneqq -H_s^{(t, y)} L_s^{(t, y)}.
\end{align*}
We can also consider $\nabla_y H^{(t, y)}$ as the derivative of $H^{(t, y)}$ with respect to $y$. When $Y$ starts from $y$ at time 0, we drop the superscripts $(0, y)$.

\begin{theorem}\label{explicit.portfolio}
Under Assumptions \ref{coeff.asp}--\ref{smooth.asp}, the optimal portfolio process $\hat{\pi}$ in Theorem~\ref{Optimal Terminal Wealth} can be represented by 
\begin{align*}
    \hat{\pi}_t &= -(\sigma^\top(Y_t))^{-1} \theta (Y_t) \frac{1}{H_t} \mathbb{E}_t \left[H_T \cdot \hat{\lambda} H_T I^\prime(\hat{\lambda}H_T) \right] \\
        & \qquad + (\sigma^\top(Y_t))^{-1}  a^\top(Y_t) \mathbb{E}_t \left[ \nabla_y H^{(t, Y_t)}_T  \left( I(\hat{\lambda H_T}) + \hat{\lambda} H_T I^\prime(\hat{\lambda} H_T) \right) \right].
\end{align*}
\end{theorem}

\begin{remark}
\begin{itemize}
    \item [(i)] We can find similar results in \cite{DGR, F, La, PS}. Here, we state the differences between our arguments and those of the previous papers. Firstly, \cite{DGR} considers the same market model as ours and assumes that $H_T I(zH_T) \in \mathbb{D}_{2, 1}$, which is difficult to check in general. In the present paper, we check that $H_T I(zH_T) \in \mathbb{D}_{1, 1}$, which is enough to use Clark's formula (Proposition~\ref{Clark.formula}), with the growth conditions for an inverse marginal utility $I$, its derivative $I^\prime$, and market coefficients (Assumptions~\ref{coeff.asp}--\ref{smooth.asp}). Secondly, \cite{La, PS} check the conditions that $H_T I(zH_T) \in \mathbb{D}_{1, 1}$ with linear Gaussian dynamics for the drift process $\mu_t$ under partial information. Their assumptions for market coefficients depend on the investment horizon $T$, but our assumptions do not depend on $T$, which is more useful for proving the turnpike theorem. Lastly, \cite{F} uses stochastic flow techniques instead of Malliavin calculus. Because the assumptions in \cite{F} also depend on $T$ and are difficult to check for our model, we do not use the results of \cite{F}.   
    \item [(ii)]  We can represent the optimal portfolio process $\hat{\pi}$ by \textit{the Arrow--Pratt measure of absolute risk tolerance} $ART_U(x) \coloneqq - \frac{U^{\prime}(x)}{U^{\prime\prime}(x)}$ as in \cite{DGR}:
    \begin{align*}
        \hat{\pi}_t &= (\sigma^\top(Y_t))^{-1} \theta (Y_t) \frac{1}{H_t} \mathbb{E}_t \left[H_T \cdot ART_U(\hat{X}_T) \right] \\
        & \qquad + (\sigma^\top(Y_t))^{-1}  a^\top(Y_t) \mathbb{E}_t \left[ \nabla_y H^{(t, Y_t)}_T  \left( \hat{X}_T - ART_U(\hat{X}_T) \right) \right].
    \end{align*}
    Moreover,  the optimal portfolio $\hat{\pi}$ can be divided into two components, namely, the myopic portfolio $\hat{\pi}^{M}$ and the excess hedging demand $\hat{\pi}^{H}$:
\begin{align*}
    \hat{\pi}_t &= \hat{\pi}^M_t + \hat{\pi}^H_t, \\
    \hat{\pi}^{M}_t &=  (\sigma^\top(Y_t))^{-1} \theta (Y_t) \frac{1}{H_t} \mathbb{E}_t \left[H_T \cdot ART_U(\hat{X}_T) \right], \\
    \hat{\pi}^{H}_t &= (\sigma^\top(Y_t))^{-1}  a^\top(Y_t) \mathbb{E}_t \left[ \nabla_y H^{(t, Y_t)}_T  \left( \hat{X}_T - ART_U(\hat{X}_T) \right) \right].
\end{align*}
\end{itemize}
\end{remark}

By the Markov property of the market model, the optimal portfolio process $\hat{\pi}$ obtained in Theorem~\ref{explicit.portfolio} is given by a feedback form. 
\begin{proposition}\label{feedback.form}
Under Assumptions \ref{coeff.asp}--\ref{smooth.asp}, the optimal portfolio process $\hat{\pi}$ in Theorem~\ref{Optimal Terminal Wealth} can be represented by a feedback form 
\[
\hat{\pi}_t = \hat{\pi}(T - t, \hat{X}_t, Y_t), \quad t \in [0, T],
\]
where $\hat{\pi}:(0, \infty) \times (0, \infty) \times \mathbb{R}^m \to \mathbb{R}^n$ is defined by
\begin{align*}    
\hat{\pi}(\tau, x, y) & \coloneqq  \hat{\pi}^M(\tau,x, y) + \hat{\pi}^H(\tau, x, y), \\
\hat{\pi}^{M}(\tau, x, y) &\coloneqq   -(\sigma^\top(y))^{-1} \theta(y) \mathbb{E}^y \left[ H_\tau \cdot {\hat{\lambda} H_\tau I^\prime(\hat{\lambda} H_\tau)}  \right], \\
\hat{\pi}^H(\tau, x, y) &\coloneqq (\sigma^\top(y))^{-1}  a^\top(y) \mathbb{E}^y \left[\nabla_y H_\tau  \left( I(\hat{\lambda} H_\tau) + \hat{\lambda} H_\tau I^\prime(\hat{\lambda} H_\tau) \right) \right],
\end{align*}
and $\hat{\lambda} = \hat{\lambda}(\tau, x, y)$ is defined by an equality $x = \mathbb{E}^y[H_\tau I(\lambda H_\tau)]$.

\end{proposition}

\subsection{Turnpike theorem for myopic portfolios under general stochastic factor models}\label{Myopic.Turnpike}
In this subsection, we consider two investors with utility functions $U_i:(0, \infty) \to (-\infty, \infty), \;(i=1, 2)$, and we fix an initial wealth $x > 0$ for both investors. We denote the corresponding optimal terminal wealth $\xi$, optimal wealth process $\hat{X}$, and optimal feedback function $\hat{\pi}$ for the $i$-th investor by $\xi^{i, T}, \hat{X}^{i, T}$, and $\hat{\pi}^{i}$, respectively.  
In this subsection, we consider the turnpike theorem for myopic portfolios under general stochastic factor models (\ref{gen.stoch.fac.model}) introduced in the previous subsection. That is, we show that for each $x$ and $y$,
\[
    |\hat{\pi}^{1, M}(T, x, y) - \hat{\pi}^{2, M}(T, x, y) | \rightarrow 0, \quad (T \nearrow \infty),
\]
 and we derive its convergence rate in terms of $\mathbb{E}[H_T]$. 

\begin{assumption}\label{pref.asp}
    \begin{enumerate}
        \item [(i)] Let $p \in (-\infty, 0]$. For $p < 0$, we set
        \[
                    U_2(x) \coloneqq \frac{x^p}{p}, \quad x \in (0, \infty),
        \]
        for $p = 0$,
        \[
            U_2(x) \coloneqq \log x, \quad x \in (0, \infty).
        \]
        \item [(ii)] Let $q \coloneqq \frac{p}{p-1} \in [0, 1)$. There exist constants $K \in [0, \infty), \alpha \in (q - 1, 0]$ such that 
        \begin{align}
            |I_1(z) - I_2(z) | &\leq K(1 + z^\alpha), \quad z \in (0, \infty), \label{I.asp} \\
            |zI^\prime_1(z) - zI^\prime_2(z) | &\leq K(1 + z^\alpha), \quad z \in (0, \infty). \label{Iprime.asp}
        \end{align}
        \item [(iii)]
        \[
            \mathbb{E}[H_T] \searrow 0, \quad (T \nearrow \infty).
        \]
    \end{enumerate}
\end{assumption}
\begin{remark}
    Because the von~Neumann--Morgenstern (vN-M) utility function $U_1$ for the first investor is uniquely determined up to positive affine transformations \cite[Theorem~2.21]{FS}, any conditions for vN-M utility function $U_1$ must be invariant for positive affine transformations. In this case, we have to assume a generalized version of Assumption~\ref{pref.asp}(ii) as follows.
\begin{assumption*}
    There exist constants $C > 0, K \geq 0, \alpha \in (q-1, 0]$ such that 
    \begin{equation}\label{gen.pref.asp1} 
        \begin{aligned}
            |I_1(z) - C I_2(z) | &\leq K(1 + z^\alpha), \quad z \in (0, \infty), \\
        |zI^\prime_1(z) - C zI^\prime_2(z) | &\leq K(1 + z^\alpha), \quad z \in (0, \infty). 
        \end{aligned}
    \end{equation}
\end{assumption*}
    If $U_1$ satisfies (\ref{gen.pref.asp1}), then $\Tilde{U}_1 \coloneqq aU_1 + b$ also satisfies (\ref{gen.pref.asp1}) for any $a > 0,\; b \in \mathbb{R}$, which means that (\ref{gen.pref.asp1}) is invariant for positive affine transformations. In particular, if $U_1$ satisfies  (\ref{gen.pref.asp1}), then $\Tilde{U}_1 \coloneqq aU_1, \; a\coloneqq C^{\frac{1}{q - 1}}$ satisfies  (\ref{gen.pref.asp1}) for $C = 1$ and the optimal feedback function given by Proposition~\ref{feedback.form} is invariant, which means that without loss of generality, we can assume $C = 1$ in  (\ref{gen.pref.asp1}), that is, (ii) in Assumption~\ref{pref.asp}.
\end{remark}

\begin{remark}
If a function $U$ is twice differentiable and satisfies (\ref{I.asp}) in Assumption~\ref{pref.asp}(ii) for $ p \in (-\infty, 1)$, then
    $U^\prime$ is regularly varying with an exponent $ p - 1$ and 
    \[
        \lim_{x \nearrow \infty} RRA_{U}(x) = \lim_{x \nearrow \infty} \frac{-U^{\prime \prime}(x) x}{U^\prime(x)} = 1 - p.
    \]
For a proof, see \cite[Proposition~2]{CH}. Furthermore, by the identity $I_2(z) = z^{q - 1}$, the ratio of the inverse marginal utilities, $\frac{I_1}{I_2}$, and the ratio of the derivatives, $\frac{I_1^\prime}{I^\prime_2}$, satisfy
\begin{align*}
    \left| \frac{I_1(z)}{I_2(z)} - 1 \right| &\leq K\left( z^{-(q - 1)} + z^{\alpha - (q - 1)} \right), \\
    \left| \frac{I^\prime_1(z)}{I_2^\prime(z)} - 1 \right| &\leq \frac{K}{1-q} \left( z^{-(q - 1)} + z^{\alpha - (q - 1)} \right). \\
\end{align*}
These inequalities imply that these ratios converge to 1, and the speed of convergence is determined by $\alpha - (q - 1)$.
\end{remark}

\begin{remark}
Although the inequality (\ref{I.asp}) can be derived from (\ref{Iprime.asp}), we assume both inequalities for simplicity.
\end{remark}

\begin{remark}
When showing turnpike theorems, Assumption~\ref{pref.asp}(iii) is the usual one in previous works. Indeed, Dybvig et~al.\ \cite{DRB} assume the same condition. Furthermore, Cox and Huang \cite{CH}, Jin \cite{J}, Huang and Zariphopoulou\cite{HZ}, and Bian and Zheng \cite{BZ, BZ19} consider the Black--Scholes model and assume that the interest rate $r$ is strictly positive, which is equivalent to $\mathbb{E}[H_T] \searrow 0$ in the model. Because 
    \[
        \mathbb{E}[H_T] = \mathbb{E}^{\mathbb{Q}}\left[ \exp{\left( - \int_0^T r(Y_t) dt \right)} \right]
    \]
    is the price of a zero-coupon bond with maturity $T$ at $t = 0$, Assumption~\ref{pref.asp}(iii) implies that the interest rate $r(Y_t)$ is positive in the long run. 
\end{remark}

The following theorem is one of our main results. 
\begin{theorem}\label{myopic.turnpike}
Under Assumptions \ref{coeff.asp}--\ref{pref.asp}, there exists an $M = M(x, y) \in ( - \infty, x]$ that is independent of $T$ such that
\begin{align*}
|\hat{\pi}^{1, M} (T, x, y) - \hat{\pi}^{2, M} (T, x, y) | &\leq K(2 - q) |\sigma^\top(y)^{-1} \theta(y)| \left( \mathbb{E}[H_T] + (x - M)^{\frac{\alpha}{q - 1}} \mathbb{E}[H_T]^{1 - \frac{\alpha}{q - 1}} \right) \\
&= O \left( \mathbb{E} [H_T]^{1 - \frac{\alpha}{q-1}} \right), \quad (T \nearrow \infty).
\end{align*}
\end{theorem}

\begin{remark}
Here we want to emphasize that the convergence rate of the turnpike theorem in stochastic factor models is determined by two components: (i) the speed of market growth, $\mathbb{E}[H_T]$, and (ii) the similarity between utilities captured by $\alpha$ and $q$. In addition, if the interest rate is a positive constant, $r(Y_t) = r > 0$, then the convergence rate is $e^{- r \left( 1 - \frac{\alpha}{q - 1} \right)T }$, which is the same rate as in Bian and Zheng \cite{BZ}. Therefore, the rate $\mathbb{E} [H_T]^{1 - \frac{\alpha}{q-1}}$ derived in the present paper is a natural extension of the rate in \cite{BZ}.
\end{remark}
\begin{remark}\label{exp.decay.ZC}
As studied by Qin and Linetsky \cite{QL1}, the decay speed of $\mathbb{E}[H_T]$ is exponential in general. 
Indeed, (ii) of Theorem~3.2 in \cite{QL1} shows that under a general semimartingale model satisfying some assumptions,
\[
    \lim_{T \nearrow \infty} \frac{- \log \mathbb{E}[H_T]}{T} = \lambda
\]
holds for some $\lambda \in \mathbb{R}$. When all uncertainty is generated by a time-homogeneous Markov process $X$, $e^{\lambda t}$ is an eigenvalue of the pricing operator $\mathcal{P}_t f(x) \coloneqq \mathbb{E}[H_t f(X_t) | X_0 = x ]$. For details, see \cite{QL1, QL2, QL3}.  
\end{remark}
Although Theorem~\ref{myopic.turnpike} seems to imply that the convergence is not uniform in $x$, we can prove uniform convergence in $x$ for portfolio proportions $\frac{\hat{\pi}^{i, M}(T, x, y)}{x}$.
\begin{theorem}\label{uni.conv.myopic}
Under Assumptions \ref{coeff.asp}--\ref{pref.asp}, for any $\epsilon > 0$,
\begin{align*}
\sup_{x > \epsilon} \left|  \frac{\hat{\pi}^{1, M} (T, x, y)}{x} - \frac{\hat{\pi}^{2, M} (T, x, y)}{x} \right| = O \left( \mathbb{E} [H_T]^{1 - \frac{\alpha}{q-1}} \right), \quad (T \nearrow \infty).
\end{align*}
\end{theorem}

\begin{remark}
    We cannot prove uniform convergence in $y$ because $y \mapsto \theta(y)$ has linear growth and is generally unbounded.
\end{remark}

We can also show that the time-0 value of the difference between the optimal wealth processes at time $t$ converges to 0 uniformly in $t$ and that the convergence rate is the same as in the above theorems.

\begin{theorem}\label{wealth.turnpike}
Under Assumptions \ref{coeff.asp}--\ref{pref.asp},
\begin{equation}\label{wealth.conv}
    \sup_{t \in [0, T]}\mathbb{E}[ H_t | \hat{X}^{1, T}_t - X^{2, T}_t |] = O \left( \mathbb{E} [H_T]^{1 - \frac{\alpha}{q-1}} \right), \quad (T \nearrow \infty).
\end{equation}
\end{theorem}

\begin{remark}
     The convergence $\lim_{T \nearrow \infty} \mathbb{E}[ H_t | \hat{X}^{1, T}_t - X^{2, T}_t |] = 0$ is already proved in complete markets by Dybvig et~al.\ \cite{DRB}. Here, our focus is on the convergence rate. 
\end{remark}

\subsection{Turnpike theorem for excess hedging demands under quadratic term structure models}\label{Excess.Turnpike}
    In this subsection, we consider quadratic term structure models studied by \cite{ADG, LW}. The market model consists of a riskless bond $S^0$, $n$ risky assets $S = (S^1, \dots, S^n)^\top$, and an $m$-dimensional factor process $Y = (Y^1, \dots, Y^m)^\top$ that affects the risk-free interest rate $r(Y_t)$ of $S^0$ and the mean return rate $\mu(Y_t)$ of $S$:
    \begin{equation}\label{QTSM1}
    \begin{aligned}
        dS^0_t &= S^0_t r(Y_t) dt,& S^0_0 &= 1, \\
        dS_t &= \mathrm{diag}(S_t) \left\{ \mu(Y_t) dt + \Sigma dW_t \right\},& S_0 &\in \mathbb{R}^n_{++},\\
        dY_t &= (b + BY_t) dt + \Lambda dW_t,& Y_0 &= y \in \mathbb{R}^m,
    \end{aligned}
    \end{equation}
    where $r:\mathbb{R}^m \to \mathbb{R},\; \mu:\mathbb{R}^m \to \mathbb{R}^n,\; \Sigma \in \mathbb{R}^{n \times n},\; b \in \mathbb{R}^m,\; B \in \mathbb{R}^{m \times m}, \Lambda \in \mathbb{R}^{m \times n}$. Furthermore, we assume that the risk-free interest rate $r(Y_t)$ of $S^0$ is a quadratic Gaussian process and the market price of risk $\theta(Y_t) \coloneqq \Sigma^{-1} \left( \mu(Y_t) - \mathbf{1} r(Y_t) \right)$ is a linear Gaussian process as follows:
    \begin{equation}\label{QTSM2}
    \begin{aligned}
        r(y) &\coloneqq r_0 + r_1^\top y + \frac{1}{2} y^\top R_2 y, \\
        \theta(y) &\coloneqq a + Ay, \\
        \mu(y) &\coloneqq \Sigma \theta(y) + \mathbf{1} r(y), 
    \end{aligned}
    \end{equation}
    where $r_0 \in \mathbb{R},\; r_1 \in \mathbb{R}^m, \; R_2 \in \mathbb{R}^{m \times m},\; a \in \mathbb{R}^n,\; A \in \mathbb{R}^{n \times m},$ and $\mathbf{1} = (1, \dots, 1)^\top \in \mathbb{R}^n$. We denote the totality of $m \times m$, real, symmetric matrices by $\mathbb{S}^m$ and $\mathbb{S}^m_+ \coloneqq \left\{M \in \mathbb{S}^m; \; M \geq 0 \right\}$, $\mathbb{S}^m_{++} \coloneqq \left\{M \in \mathbb{S}^m; \; M > 0 \right\}$. We assume the following conditions.
    \begin{assumption}\label{LQMarket.asp} \label{stable.asp}
        \begin{itemize}
            \item [(i)] $\Sigma $ is invertible.
            \item [(ii)] $R_2 \in \mathbb{S}^m_+$.
            \item [(iii)] $R_2 = 0$ or $\left( \gamma ( 1- \gamma) A^\top A + \gamma R_2 \right) \in \mathbb{S}^m_{++}$ for $\gamma \in \left\{q, 1  + \alpha \right\}$.
            \item [(iv)] $B$ is stable; that is, all its eigenvalues have negative real parts.
        \end{itemize}
    \end{assumption}
    Under Assumptions \ref{pref.asp} and \ref{LQMarket.asp}, the quadratic term structure model given by (\ref{QTSM1}) and (\ref{QTSM2}) satisfies Assumptions \ref{coeff.asp}--\ref{smooth.asp}. Therefore, all main results in Sects.~\ref{flow.rep.sec} and \ref{Myopic.Turnpike} are valid.
    \begin{remark}
    Assumption~\ref{LQMarket.asp}(iv) implies that $Y$ is a multivariate Ornstein--Uhlenbeck process. In particular, the model includes well-known short-rate models, such as the Vasicek model and special versions of the CIR model; for details, see \cite{ADG}.
    \end{remark}
    \begin{remark}
        We restrict our analysis to the quadratic term structure model given by (\ref{QTSM1}) and (\ref{QTSM2}) for the following reason. As Proposition~\ref{excess.hedging.bdd} says, to estimate the rates for excess hedging demands, we have to compute the asymptotic behavior of the stochastic factor process $Y$ under myopic probability measures $\mathbb{Q}^\gamma_T, \;  \gamma \in [0, 1],$ defined by 
        \[
            d\mathbb{Q}^\gamma_T \coloneqq \frac{H_T^\gamma}{\mathbb{E}[H_T^\gamma]}d\mathbb{P}.
        \]
        When $\gamma = 1$, the myopic probability measure $\mathbb{Q}^1_T$ is a $T$-forward measure. When $\gamma \in ( 0, 1)$, the measures are given by the optimal wealth processes $\hat{X}^T_T$ for CRRA investors:
        \[
            d\mathbb{Q}^\gamma_T = \frac{U(\hat{X}^{T}_T)}{\mathbb{E}\left[U(\hat{X}^{T}_T)\right]} d\mathbb{P},      
        \]
        where, $U(x) \coloneqq \frac{x^p}{p}, \; p \coloneqq \frac{\gamma}{\gamma - 1}$. As Guasoni et~al.\ \cite{GKRX} say, by using results on CRRA utility maximization problems, we can represent the probability density processes of $\mathbb{Q}^\gamma_T$ as stochastic exponential martingales in terms of the optimal portfolio processes. Therefore, in the quadratic term structure model given by (\ref{QTSM1}) and (\ref{QTSM2}), we can analyze the asymptotic behavior of $Y$ under $\mathbb{Q}^\gamma_T$ by using the asymptotic behavior of the solutions to Riccati differential equations. Because the optimal portfolios in nonlinear stochastic factor models (\ref{gen.stoch.fac.model}) are given by solutions of semilinear PDEs (see Nagai \cite{Na}), our analysis will require more-advanced techniques, and further research for nonlinear stochastic factor models will be addressed in future work. 
    \end{remark}

Under these assumptions, we can also derive the convergence rates for excess hedging demands.
\begin{theorem}\label{excess.turnpike}
Under Assumptions \ref{pref.asp} and \ref{LQMarket.asp},
\[
    | \hat{\pi}^{1, H} (T, x, y)  - \hat{\pi}^{2, H} (T, x, y)| = O \left( \mathbb{E} [H_T]^{1 - \frac{\alpha}{q - 1}} \right), \quad (T \nearrow \infty).
\]
\end{theorem}

By combining the above theorem with Theorem~\ref{myopic.turnpike}, we obtain the convergence rates for the optimal feedback functions $\hat{\pi}^1 = \hat{\pi}^{1, M} + \hat{\pi}^{1, H}$.
\begin{corollary}
    Under Assumptions \ref{pref.asp} and \ref{LQMarket.asp},
\[
    | \hat{\pi}^{1} (T, x, y)  - \hat{\pi}^{2} (T, x, y)| = O \left( \mathbb{E} [H_T]^{1 - \frac{\alpha}{q - 1}} \right), \quad (T \nearrow \infty).
\]
\end{corollary}
Moreover, by considering portfolio proportions rather than dollar amounts, we can prove the uniform turnpike theorem for excess hedging demands and obtain its convergence rate. 
\begin{theorem}\label{uni.excess.turnpike}
Under Assumptions \ref{pref.asp} and \ref{LQMarket.asp}, for any $\epsilon > 0$, 
    \[
    \sup_{x > \epsilon} \left| \frac{\hat{\pi}^{1, H} (T, x, y)}{x}  - \frac{\hat{\pi}^{2, H} (T, x, y)}{x} \right| = O \left( \mathbb{E} [H_T]^{1 - \frac{\alpha}{q - 1}} \right), \quad (T \nearrow \infty).
\]
\end{theorem}
\begin{corollary}
Under Assumptions \ref{pref.asp} and \ref{LQMarket.asp}, for any $\epsilon > 0$, 
    \[
    \sup_{x > \epsilon} \left| \frac{\hat{\pi}^{1} (T, x, y)}{x}  - \frac{\hat{\pi}^{2} (T, x, y)}{x} \right| = O \left( \mathbb{E} [H_T]^{1 - \frac{\alpha}{q - 1}} \right), \quad (T \nearrow \infty).
\]
\end{corollary}

\subsection{Applications: optimal collective investment problems}\label{Application.sec}
In this subsection, we offer applications of our main results to optimal collective investment problems. For detailed descriptions of the problems, see \cite{BCMN, JN}, for example. We consider the quadratic term structure model, given by (\ref{QTSM1}) and (\ref{QTSM2}), that satisfies Assumptions \ref{pref.asp}(iii) and \ref{LQMarket.asp} as described in Sect.~\ref{Excess.Turnpike}. We assume there are $n$ investors with CRRA utility functions $U_i(x) \coloneqq \frac{x^{p_i}}{p_i},\; (i = 1, \dots, n), \; -\infty < p_1 < \dots p_n \leq 0$, where the relative risk aversion level of the $i$-th investor is given by $\gamma_i \coloneqq 1-p_i$. At the beginning of the investment period ($t = 0$), investors delegate their investment management to a fund manager. At the end of the investment period ($t = T$), a fund manager distributes the aggregate terminal wealth among $n$ investors according to a sharing rule. In this paper, we consider two well-known sharing rules: a Pareto optimal sharing rule and a linear sharing rule.
\subsubsection{Pareto optimal sharing rule}
We assume that a fund manager chooses a Pareto optimal distribution of the terminal wealth, which is represented by 
\[
    \Tilde{U}(x) \coloneqq \max_{} \left\{ \sum_{i = 1}^n \beta_i U_i (x_i) \left| \; x_i \in \mathbb{R},\;   \sum_{i = 1}^n x_i = x \right. \right\},
\]
where $\beta_i$ is the weight granted to the $i$-th investor and satisfies $\beta_i > 0, \; (i = 1, \dots, n),\; \sum_{i = 1}^n \beta_i = 1$. \cite{BCMN} uses this utility function to analyze collective investment problems. $\Tilde{U}$ also appears as a utility function for a representative agent in the context of market equilibrium \cite[Chapter 4]{KS}. 
Because the operations of sup-convolution and addition are dual to each other, the inverse marginal utility of a fund manager is given by the sum of those of individuals.
\begin{lemma}
    $\Tilde{U}$ is a utility function in Definition \ref{Util.def} and the inverse marginal utility $\Tilde{I} \coloneqq (\Tilde{U})^{-1}$ is given by 
    \[
        \Tilde{I}(z) = \sum_{i = 1}^n I_i \left( \frac{z}{\beta_i} \right), \quad I_i (z) = z^{q_i - 1}, \quad q_i \coloneqq \frac{p_i}{p_i - 1}.
    \]
\end{lemma}
\begin{proof}
    See \cite[Appendix~A]{BCMN}.
\end{proof}
By using our main results in Sect.~\ref{Excess.Turnpike}, we can show that the optimal feedback function of a fund manager converges to that of the most risk-seeking investor. The convergence rate is determined by the price of a zero-coupon bond and the relative difference between the relative risk aversion levels of the least risk-averse investor ($i = n$) and the second least risk-averse investor ($i = n - 1$). 
\begin{theorem}\label{Pareto.tunpike}
    Let $\hat{\pi}^P$ be the optimal feedback function of a fund manager with the Pareto optimal sharing rule and $\hat{\pi}^i$ be that of the $i$-th investor. Then
    \[
        \left| \hat{\pi}^P(T, x, y) - \hat{\pi}^n (T, x, y) \right| = O\left( \mathbb{E}[H_T]^{\frac{\gamma_{n-1} - \gamma_n}{\gamma_{n-1}}} \right), \quad (T \nearrow \infty).
    \]
\end{theorem}
\begin{proof}
    It is enough to check that the inequalities (\ref{gen.pref.asp1}) hold for $I_1(z) = \Tilde{I}(z), \; I_2(z) = z^{q_n - 1},\; \alpha = q_{n-1} - 1, \; C = \left( \frac{1}{\beta_n} \right)^{q_n - 1}$. Because $0 \leq q_n < q_{n-1} < \dots < q_1 < 1$, there exists some $K > 0$ such that for all $z > 0$, 
    \begin{align*}
        \left| \Tilde{I}(z) - \left( \frac{z}{\beta_n} \right)^{q_{n} - 1} \right| &= \sum_{i = 1}^{n-1} \left( \frac{z}{\beta_i} \right)^{q_i - 1} \\
        &\leq K(1 + z^{q_{n-1} - 1})
    \end{align*}
    and 
    \begin{align*}
        \left| z\Tilde{I}^\prime(z) - (q_n - 1) \left( \frac{z}{\beta_n} \right)^{q_{n} - 1} \right| &= \sum_{i = 1}^{n-1} (1 - q_i)\left( \frac{z}{\beta_i} \right)^{q_i - 1} \\
        &\leq K(1 + z^{q_{n-1} - 1})
    \end{align*}
    hold. Moreover, the identity
    \begin{align*}
       \frac{\alpha}{q - 1} =  \frac{q_{n-1} - 1}{q_n - 1} = \frac{1 - p_n}{1 - p_{n-1}} = \frac{\gamma_n}{\gamma_{n-1}}
    \end{align*}
    gives the convergence rate.
\end{proof}

\subsubsection{Linear sharing rule}
Here, we assume that a fund manager allocates the terminal wealth according to a linear sharing rule, where the $i$-th investor receives a fixed proportion $\alpha_i$ of the terminal wealth. This linear sharing rule is represented by 
\[
    U(x) \coloneqq  \sum_{i = 1}^n \beta_i U_i (\alpha_i x),
\]
where $\alpha_i, \;\beta_i > 0, \; (i = 1, \dots, n),\; \sum_{i = 1}^n \alpha_i = \sum_{i = 1}^n \beta_i = 1$. 
Although $U$ seems simpler than $\Tilde{U}$, we cannot derive the inverse marginal utility of $U$ in general. Therefore, to estimate the differences between inverse marginal utilities and their derivatives, we perform complex calculations; for the proof, see Sect.~\ref{proof.appl.sec}.
\begin{proposition}\label{lin.share.diffI}Let 
\begin{align*}
    U_1(x) &\coloneqq \sum_{i = 1}^n \beta_i \frac{(\alpha_i x)^{p_i}}{p_i} = \sum_{i=1}^n w_i \frac{x^{p_i}}{p_i}, \\
    U_2(x) &\coloneqq w_n \frac{x^{p_n}}{p_n},
\end{align*}
where $w_i \coloneqq \beta_i \alpha_i^{p_i} > 0$.
For any nonnegative $\beta \in (1 + p_{n-1} - p_n, 1)$, there exists $K > 0$ such that for all $z > 0$,
\begin{align}
    |I_1(z) - I_2(z)| &\leq K(1 + z^{\beta (q_n - 1)}), \label{lin.share.diffI1} \\
    |zI_1^\prime(z) - z I_2^\prime(z)| &\leq K (1 + z^{\beta (q_n - 1)}) \label{lin.share.diffI2}
\end{align}
hold.
\end{proposition}
\begin{theorem}\label{Linshre.turnpike}
    Let $\hat{\pi}^L$ be the optimal feedback function of a fund manager with the linear sharing rule and $\hat{\pi}^i$ be that of the $i$-th investor. Then for any nonnegative $\beta \in (1 +p_{n-1} - p_n, 1)$,
    \[
        \left| \hat{\pi}^L(T, x, y) - \hat{\pi}^n (T, x, y) \right| = O\left( \mathbb{E}[H_T]^{1 - \beta} \right), \quad (T \nearrow \infty)
    \]
    holds.
\end{theorem}
\begin{proof}
    By Proposition~\ref {lin.share.diffI} and the main results in Sects.~\ref{Myopic.Turnpike} and \ref{Excess.Turnpike}, the convergence rate is given by 
    \[
    \mathbb{E}[H_T]^{1 - \frac{\beta(q_n - 1)}{q_n - 1}} = \mathbb{E}[H_T]^{1 - \beta}.
    \]
\end{proof}
\begin{remark}
If $1 + p_{n-1} - p_n \geq 0$, which is equivalent to $\gamma_{n-1} \leq \gamma_n + 1$ in terms of relative risk aversion levels $\gamma_i$,  then the convergence rate is given by $O(\mathbb{E}[H_T]^{\gamma_{n-1} - \gamma_n - \epsilon})$ for any $\epsilon > 0$. If $1 + p_{n-1} - p_n < 0$,  which is equivalent to $\gamma_{n-1} > \gamma_n + 1$, then the convergence rate is given by $O(\mathbb{E}[H_T])$. These facts mean that for a fund manager with a linear sharing rule, the convergence rate is determined by $\gamma_{n-1} - \gamma_n$, the absolute difference between the relative risk aversion levels of the least risk-averse investor and the second least risk-averse investor. On the other hand, for a fund manager with a Pareto optimal sharing rule, the convergence rate is determined by $\frac{\gamma_{n-1} - \gamma_n}{\gamma_{n-1}}$, the relative difference between the relative risk aversion levels for the same investors. In particular, because we consider the case $\gamma_{n-1} > 1$, the convergence rate under a linear sharing rule is faster than that under a Pareto optimal sharing rule.  
\end{remark}

\section{Proofs for main results}\label{proof.section}
\subsection{Proofs for Sect.~\ref{flow.rep.sec}}\label{flow.rep.proof}
The following lemma is used in the proof of Theorem~\ref{Optimal Terminal Wealth} and Lemma~\ref{HinD11}, which shows $H_T \in \mathbb{D}_{1, 1}$.

\begin{lemma}\label{fin.mom.beta}
Under Assumptions \ref{coeff.asp} and \ref{martingale.method.asp}, for any $T > 0,\; \lambda \geq 0,\; y \in \mathbb{R}^m$,
    \[
        \mathbb{E}^{\mathbb{Q}}_y \left[ \exp \left( - \lambda \int_0^T r(Y_t) dt \right) \right] < \infty.
    \]
\end{lemma}
\begin{proof}
    We fix $T > 0,\; y \in \mathbb{R}^m$ and for any $\lambda \geq 0$, 
    \begin{align*}
        \mathbb{E}^{\mathbb{Q}}_y \left[ \exp \left( - \lambda \int_0^T r(Y_t) dt \right) \right] &\leq \mathbb{E}^{\mathbb{Q}}_y \left[ \frac{1}{T} \int_0^T \exp \left( -\lambda T r(Y_t) \right) dt  \right] \\
        &\leq \sup_{t \in [0, T]} \mathbb{E}^{\mathbb{Q}}_y \left[ \exp \left( -\lambda T r(Y_t) \right) \right] \\
        &\leq \sup_{t \in [0, T]} \mathbb{E}^{\mathbb{Q}}_y \left[ \exp \left( -\lambda T \left( r_0 + r_1^\top Y_t \right) \right) \right] \\
        &= \exp \left( - \lambda T r_0 \right) \sup_{t \in [0, T]} \mathbb{E}^{\mathbb{Q}}_y \left[ \exp \left( -\lambda T r_1^\top Y_t  \right) \right]
    \end{align*}
holds, where the first inequality follows from Jensen's inequality, and the third inequality follows from Assumption~\ref{martingale.method.asp}(i). Therefore, it suffices to prove that for any $\lambda \geq 0$,
\[
    \sup_{t \in [0, T]} \mathbb{E}^{\mathbb{Q}}_y \left[ \exp \left( \lambda |Y_t|  \right) \right]
\]
holds. Because $Y$ satisfies the following SDE under $\mathbb{Q}$, 
\[
Y_t = y + \int_0^t \Tilde{b}(Y_s) ds + \int_0^t a(Y_s) dW^\mathbb{Q}_s,
\]
and $\Tilde{b}(y) = b(y) - a(y) \theta(y)$ is of linear growth, 
$|Y_t|$ satisfies 
\[
|Y_t| \leq |y| + KT + \sup_{t \in [0, T]} |M_t| + \int_0^t K|Y_s|ds,
\]
where $K$ is a some constant and $M_t \coloneqq \int_0^t a(Y_s) dW^{\mathbb{Q}}_s$. By Gronwall's inequality, 
\[
    |Y_t| \leq \left( |y| + KT + \sup_{t \in [0, T]}|M_t| \right) \exp(KT)
\]
holds, which leads to 
\begin{align*}
    \exp \left( \lambda |Y_t| \right) &\leq \exp \left\{ \lambda \left( |y| + KT + \sup_{t \in [0, T]}|M_t| \right) \exp(KT) \right\} \\
    &= C_1 \exp \left( C_2  \sup_{t \in [0, T]} |M_t| \right) \quad \left( C_1 \coloneqq \exp\{ \lambda (|y| + KT) \exp(KT)\}, \quad C_2 \coloneqq \lambda \exp(KT) \right)\\
    &\leq C_1 \exp\left( C_2 \sum_{i = 1}^m \sup_{t \in [0, T]} |M^i_t| \right) \\
    &\leq \sum_{i = 1}^m C_1 \exp\left( mC_2 \sup_{t \in [0, T]} |M^i_t| \right).
\end{align*}
As a result, it suffices to show that for any $\lambda \geq 0, \; i = 1, \dots, m$,
\[
    \mathbb{E}^{\mathbb{Q}}_y \left[ \exp \left( \lambda \sup_{t \in [0, T]} |M^i_t|  \right) \right] < \infty
\]
holds. The Dambis--Dubins--Schwarz theorem implies that there exists a Brownian motion $\beta^i$ such that $M^i_t = \beta^i_{\langle M\rangle_t}$, and because $a$ is bounded [Assumption~\ref{coeff.asp}(iv)], there exists a constant $L$ such that $\langle M \rangle_T \leq LT$ holds. Therefore, 
\begin{align*}
        \mathbb{E}^{\mathbb{Q}}_y \left[ \exp \left( \lambda \sup_{t \in [0, T]} |M^i_t|  \right) \right] &= \mathbb{E}^{\mathbb{Q}}_y \left[ \exp \left( \lambda \sup_{t \in [0, T]} |\beta^i_{\langle M \rangle_t}|  \right) \right] \\
        &\leq \mathbb{E}^{\mathbb{Q}}_y \left[ \exp \left( \lambda \sup_{t \in [0, LT]} |\beta^i_t|  \right) \right] \\
        &< \infty,
\end{align*}
which completes the proof.
\end{proof}

\begin{proof}[Proof of Theorem~\ref{Optimal Terminal Wealth}]
    By \cite[Theorem~3.7.6]{KS}, it suffices to check that
    \[
        \mathbb{E}[H_T] < \infty, \quad \mathbb{E}[H_T I(zH_T)]
    \]
    for any $z > 0$. By Lemma~\ref{fin.mom.beta}, 
    \begin{align*}
        \mathbb{E}[H_T] = \mathbb{E}^{\mathbb{Q}}\left[ \exp\left( - \int_0^T r(Y_t) dt \right) \right] < \infty.
    \end{align*}
    By Assumption~\ref{martingale.method.asp}(ii), 
    \begin{align*}
        H_T I(zH_T) \leq  \kappa \left( H_T + z^{-\rho} H_T^{1 - \rho}  \right) 
    \end{align*}
    holds, and by Jensen's inequality, 
    \begin{align*}
         \mathbb{E}[H_T I(zH_T)] &\leq  \kappa \left( \mathbb{E}[H_T] + z^{-\rho} \mathbb{E}\left[H_T^{1 - \rho}\right]  \right) \\
         &\leq \kappa \left( \mathbb{E}[H_T] + z^{-\rho} \mathbb{E}\left[H_T\right]^{1 - \rho} \right) \\
         &< \infty
    \end{align*}
    holds, which completes the proof.
\end{proof}
Throughout this subsection, we assume Assumptions \ref{coeff.asp}--\ref{smooth.asp}.
From Proposition~\ref{diff.sol.SDE} and Assumption~\ref{smooth.asp}(ii), the following holds.
\begin{lemma}\label{Factor.diff}
$Y = (Y^1, \dots, Y^m)$ satisfies the following.
\begin{itemize}
    \item [(i)] $Y^k_s \in \bigcap_{p \geq 1} \mathbb{D}_{p, 1}, \quad k = 1, \dots, m, \; s \in [0, T]$.
    \item [(ii)] $D_tY_s$ satisfies 
    \[
        D_t Y_s = a(Y_t) + \int_t^s Db(Y_u) D_t Y_u du + \sum_{l = 1}^n \int_t^s D a_{\cdot l}(Y_u) D_t Y_u dW^l_u
    \]
    for $t \in [0, s]$ and $D_t Y_s = 0$ for $t \in (s, T]$.
    \item [(iii)]  For $j = 1, \dots, n, \; p \in [1, \infty),$
        \[
            \sup_{r \in [0, T] } \mathbb{E} \left[ \sup_{s \in [0, T]} \left| D^j_r Y^k_s \right|^p \right] < \infty.
        \]

    \item [(iv)] $D_t Y_s = \nabla_y Y_s (\nabla_y Y_t)^{-1} a(Y_t)$ for $t \in [0, s]$, where $\nabla_y Y$ is an $\mathbb{R}^{m \times m}$-valued stochastic process satisfying
    \[
        \nabla_y Y_s = I + \int_0^s Db(Y_u) \nabla_y Y_u du + \sum_{j = 1}^n \int_t^s D a_{\cdot j}(Y_u) \nabla_y Y_u dW^j_u
    \]
    for $s \in [0, T]$ and $I \in \mathbb{R}^{m \times m}$ is the identity matrix. 
\end{itemize}

\end{lemma}

\begin{lemma}
    For $s, t \in [0, T]$,
    \begin{itemize}
        \item [(i)] $D_s (r(Y_t)) = Dr(Y_t) D_sY_t$,
        \item [(ii)] $D_s (\theta(Y_t)) = D\theta(Y_t) D_s Y_t$,
        \item [(iii)] $D_s \left( \frac{1}{2}| \theta(Y_t)|^2 \right) = \theta^\top(Y_t) D\theta(Y_t) D_s Y_t$,
        \item [(iv)] $D_s \int_0^T r(Y_t) dt = \int_s^T Dr(Y_t) D_sY_t dt$,
        \item [(v)] $D_s \left( \frac{1}{2} \int_0^T | \theta(Y_t)|^2 dt \right) = \int_s^T \theta^\top(Y_t) D\theta(Y_t) D_sY_t dt$,
        \item [(vi)] $D_s\int_0^T \theta^\top(Y_t) dW_t = \theta^\top(Y_s) + \left(  \int_s^T ( D\theta(Y_t) D_sY_t)^\top dW_t \right)^\top$. 
    \end{itemize}
\end{lemma}
\begin{proof}
    (i), (ii), and (iii) follow from Proposition~\ref{chain.rule} with the remark below it, Lemma~\ref{Factor.diff}, and Assumption~\ref{smooth.asp}(i) that $r, \theta, Dr, D\theta$ are of polynomial growth.  (iv) and (v) follow from Proposition~\ref{diff.under.leb}, and (vi) follows from Proposition~\ref{diff.under.SI}.
\end{proof}
The following lemma implies finiteness for $p$-th moments of $\sup_{t \in [0, T]} |Y_t|$ and $|D_t Y_s|$ under the equivalent martingale measure $\mathbb{Q}$. In this lemma, we use (iii) of Assumption~\ref{smooth.asp}, which is a monotone condition for SDE (\ref{SDE.YY.Q}) that $(Y_s, D^1_tY_s, \dots D^n_tY_s)_{s \in [t, T]}$ satisfies under $\mathbb{Q}$.
\begin{lemma}\label{Y.DY.moment.Q}
For any $p \in [1, \infty)$,
    \begin{enumerate}
        \item [(i)] $\mathbb{E}^{\mathbb{Q}}[\sup_{s \in [0, T]}|Y_s|^p] < \infty,$
        \item [(ii)] $\sup_{s, t \in [0, T]} \mathbb{E}^{\mathbb{Q}} [|D_tY_s|^p] < \infty.$ 
    \end{enumerate}
\end{lemma}
\begin{proof}
    \begin{enumerate}
        \item [(i)] $Y = (Y_t)_{t \in [0, T]}$ satisfies
        \begin{equation}\label{SDE.Y.Q}
        \begin{aligned}
            dY_s &= b(Y_s) ds + a (Y_s) dW_s \\
                 &= \left\{ b(Y_s) - a(Y_s) \theta(Y_s) \right\} ds + a(Y_s) dW^\mathbb{Q}_s.
        \end{aligned}
        \end{equation}
        Because $b$ and $\theta$ have linear growth and $a$ is bounded, SDE (\ref{SDE.Y.Q}) satisfies a linear growth condition and (i) follows from Theorem~4.4 in \cite[Chapter~2]{Mao}.
        \item [(ii)] Let $t \in [0, T]$. 
        $\mathbb{R}^m$-valued stochastic processes $(D^j_tY_s)_{s \in [t, T]}, \; j = 1, \dots, n$ satisfy
        \begin{equation*}
            \begin{aligned}
                D^j_tY_s &= a_{\cdot j}(Y_t) + \int_t^s Db(Y_u) D^j_tY_u du + \sum_{l=1}^n \int_t^s Da_{\cdot l} (Y_u) D^j_tY_u dW^l_u \\
                &= a_{\cdot j}(Y_t) + \int_t^s \left( Db(Y_u) - \sum_{l=1}^n  Da_{\cdot l} (Y_u) \theta^l(Y_u)  \right) D^j_tY_u du + \sum_{l=1}^n \int_t^s Da_{\cdot l} (Y_u) D^j_tY_u dW^{\mathbb{Q}, l}_u.
            \end{aligned}
        \end{equation*}
        Set $\mathbb{Y} = (Y_s, D^1_tY_s, \dots D^n_tY_s)^\top_{s \in [t, T]}$. Then, $\mathbb{R}^{m(n+1)}$-valued stochastic process $\mathbb{Y}$ satisfies 
        \begin{equation}\label{SDE.YY.Q}
            d\mathbb{Y}_s = B(\mathbb{Y}_s) ds + \sum_{l=1}^n A^l(\mathbb{Y}_s) dW^{\mathbb{Q}, l}_s,
        \end{equation}
        with $\mathbb{Y}_t = (Y_t, a_{\cdot 1}(Y_t), \dots,  a_{\cdot n}(Y_t))^\top$, where $B, A^l:\mathbb{R}^{m(n+1)} \to \mathbb{R}^{m(n+1)}, \; l = 1, \dots, n$ are given by
        \begin{align*}
            B(y, z_1, \dots, z_n) = \begin{pmatrix}
              b(y) - a(y) \theta(y) \\
              \left( Db(y) - \sum_{l = 1}^n Da_{\cdot l}(y) \theta^l(y) \right) z_1 \\
              \vdots \\
              \left( Db(y) - \sum_{l = 1}^n Da_{\cdot l}(y) \theta^l(y) \right) z_n 
            \end{pmatrix}, \quad A^l(y, z_1, \dots, z_n) = \begin{pmatrix}
                a_{\cdot l}(y) \\
                Da_{\cdot l}(y) z_1 \\
                \vdots \\
                Da_{\cdot l}(y) z_n \\
            \end{pmatrix}
        \end{align*}
        for $y, z_1 \dots, z_n \in \mathbb{R}^m$. By Assumption~\ref{smooth.asp}(iii), SDE (\ref{SDE.YY.Q}) satisfies a monotone condition, which leads to 
        \[
            \mathbb{E}^{\mathbb{Q}}[|D_sY_t|^p] \leq C (1 + \mathbb{E}^{\mathbb{Q}}[|Y_t|^p])
        \]
        from Theorem~4.1 in \cite[Chapter~2]{Mao}, where $C > 0$ does not depend on $s, t$. Using (i) in this lemma, we complete the proof.
    \end{enumerate}
\end{proof}

\begin{lemma}\label{L.tilde.mom}
Let $\Tilde{L} = (\Tilde{L}_s)_{s \in [0, T]}$ be given by
\[
    \Tilde{L}_s \coloneqq \int_s^T Dr(Y_t) D_sY_t dt +  \theta^\top(Y_s) + \left(  \int_s^T ( D\theta(Y_t) D_sY_t)^\top dW_t \right)^\top +  \int_s^T \theta^\top(Y_t) D\theta(Y_t) D_sY_t dt , \quad s \in [0, T].
\]
Then 
\[
    \sup_{s \in [0, T]} \mathbb{E}^{\mathbb{Q}}[|\Tilde{L}_s|^2] < \infty.
\]
\end{lemma}
\begin{proof}
    Using the assumption that $Dr, \theta, D\theta$ are of polynomial growth and Lemma~\ref{Y.DY.moment.Q}, we can easily prove this lemma. 
\end{proof}

\begin{lemma}\label{HinD11}
    \begin{enumerate}
        \item [(i)] $H_T \in \mathbb{D}_{1, 1}$ and 
        \begin{align*}
            D_sH_T &= - H_T \left( \int_s^T Dr(Y_t) D_sY_t dt +  \theta^\top(Y_s) + \left(  \int_s^T ( D\theta(Y_t) D_sY_t)^\top dW_t \right)^\top +  \int_s^T \theta^\top(Y_t) D\theta(Y_t) D_sY_t dt \right) \\
            & = -H_T \Tilde{L}_{s}.
        \end{align*}
        \item [(ii)] Let $z > 0$. $H_T I(zH_T) \in \mathbb{D}_{1, 1}$ and 
        \begin{align*}
            D_s \left( H_T I(z H_T) \right) = -H_T \Tilde{L}_{s} (I(z H_T) + z H_T I^\prime(zH_T)).
        \end{align*}
    \end{enumerate}
\end{lemma}
\begin{proof}
    \begin{itemize}
        \item [(i)] From Lemmas \ref{fin.mom.beta} and \ref{L.tilde.mom},
        \begin{align*}
            \mathbb{E}\left[ \left( \int_0^T |H_T \Tilde{L}_s|^2 ds \right)^{\frac{1}{2}} \right] &= \mathbb{E}^{\mathbb{Q}} \left[ \beta_T \left( \int_0^T |\Tilde{L}_s|^2 ds \right)^{\frac{1}{2}}   \right] \\
            &\leq  \mathbb{E}^{\mathbb{Q}} \left[ \beta_T^2 \right]^{\frac{1}{2}} \left( \int_0^T \mathbb{E}^{\mathbb{Q}} \left[  |\Tilde{L}_s|^2\right] ds \right)^{\frac{1}{2}} \\
            &\leq  \mathbb{E}^{\mathbb{Q}} \left[ \beta_T^2 \right]^{\frac{1}{2}}  \left( T \cdot \sup_{s \in [0, T]}\mathbb{E}^{\mathbb{Q}} \left[  |\Tilde{L}_s|^2\right]  \right)^{\frac{1}{2}} \\
            &< \infty
        \end{align*}
        holds. Therefore, we can use the chain rule (Proposition~\ref{chain.rule}), which proves (i).
        \item [(ii)] By Assumptions \ref{martingale.method.asp}(ii) and \ref{smooth.asp}(iv), there exists $\kappa > 0$ such that 
        \[
            I(z) + z I^\prime(z) \leq \kappa (1 + z^{-1})
        \] for any $z > 0$. As a result, 
        \begin{align*}
            |H_T \Tilde{L}_s(I(zH_T) + z H_T I^\prime(zH_T))| \leq \kappa H_T |\Tilde{L}_s| + \frac{\kappa}{z} |\Tilde{L}_s|,
        \end{align*}
        and
        \begin{align*}
            \mathbb{E}\left[ \left( \int_0^T |H_T \Tilde{L}_s|^2 ds \right)^{\frac{1}{2}} \right] < \infty,  \quad \mathbb{E}\left[ \left( \int_0^T |\Tilde{L}_s|^2 ds \right)^{\frac{1}{2}} \right] < \infty
        \end{align*}
        holds. Therefore, we can use the chain rule again and complete the proof.
    \end{itemize}
\end{proof}

Let a pair $(Y,Z)$ of $\mathbb{R}^m$-valued stochastic process $Y$ and $\mathbb{R}^{m \times m}$-valued stochastic process $Z$ be the solution to the following system of SDEs:
\begin{align*}
    dY_s &= b(Y_s) ds + a(Y_s) dW_s,& Y_t &= y, \\
    dZ_s &= Db(Y_s) Z_s ds + \sum_{j = 1}^n D a_{\cdot j}(Y_s) Z_s dW^j_s,& Z_t &= I. 
\end{align*}
Then $(Y, Z)$ is a Markov process and $(Y^{(t, y)}, Z^{(t, y)})$ denotes the solution to the above system of SDEs when $(Y, Z)$ starts from $(y, I) \in \mathbb{R}^{m} \times \mathbb{R}^{m \times m}$. Note that $Z^{(t, y)}$ always starts from the identity matrix $I \in \mathbb{R}^{m \times m}$ and $Y^{(0, y)}_s = Y^{(t, Y_t^{(0, y)})}_s$ for $s \in [t, T]$. Because $Z^{(t, y)}$ can be thought of as the derivative of $Y^{(t, y)}$ with respect to the initial value $y$, we use the notation $\nabla_y Y^{(t, y)} \coloneqq Z^{(t, y)}$ instead of $Z^{(t, y)}$. Using these notations and Lemma~\ref{Factor.diff}, 
\[
    D_tY_s = \nabla_y Y_s^{(t, Y_t)} a(Y_t), \quad s \in [t, T],
\]
where $Y_t = Y^{(0, y)}_t$.
Furthermore, let $H^{(t, y)} = \left( H_s^{(t, y)} \right)_{s \in [t, T]} $ be
\[
    H_s^{(t, y)} \coloneqq  \exp \left( - \int_t^s r(Y^{(t, y)}_u) du - \int_t^s \theta^\top (Y^{(t, y)}_u) dW_u - \frac{1}{2} \int_t^s |\theta (Y^{(t, y)}_u)|^2 du \right),
\]
and let $\nabla_y H^{(t, y)} = \left( \nabla_y H_s^{(t, y)} \right)_{s \in [t, T]}, \; L^{(t, y)} = \left( L_s^{(t, y)}\right)_{s \in [t, T]} $ be an $\mathbb{R}^m$-valued stochastic process given by 
\begin{align*}
    L_s^{(t, y)} &\coloneqq \int_t^s (Dr(Y^{(t, y)}_u) \nabla_y Y^{(t, y)}_u )^\top du + \int_t^s \left( D\theta (Y^{(t, y)}_u) \nabla_y Y^{(t, y)}_u \right)^\top  dW_u + \int_t^s \left( D\theta (Y^{(t, y)}_u) \nabla_y Y_u^{(t, y)} \right)^\top \theta(Y_u^{(t, y)}) du, \\
     \nabla_y H_s^{(t, y)} &\coloneqq -H_s^{(t, y)} L_s^{(t, y)}.
\end{align*}
From these notations,
\begin{align*}
    D_tH_T &= -H_T \Tilde{L}_t \\
           &= - H_T \left( \int_t^T Dr(Y_s) D_tY_s ds +  \theta^\top(Y_t) + \left(  \int_t^T ( D\theta(Y_s) D_tY_s)^\top dW_s \right)^\top +  \int_s^T \theta^\top(Y_s) D\theta(Y_s) D_tY_s ds \right) \\
           &= -H_T \theta^\top (Y_t) -H_t^{(0, y)} H_T^{(t, Y_t)} \left( \int_t^T Dr(Y_s) D_tY_s ds +  \left(  \int_t^T ( D\theta(Y_s) D_tY_s)^\top dW_s \right)^\top +  \int_s^T \theta^\top(Y_s) D\theta(Y_s) D_tY_s ds \right) \\
           &\begin{multlined}
              =-H_T \theta^\top (Y_t) - H_t^{(0, y)} H_T^{(t, Y_t)} \left( \int_t^T Dr(Y_s)  \nabla_y Y_s^{(t, Y_t)} ds +  \left(  \int_t^T ( D\theta(Y_s)  \nabla_y Y_s^{(t, Y_t)} )^\top dW_s \right)^\top \right. \\
                 \left. +  \int_s^T \theta^\top(Y_s) D\theta(Y_s)  \nabla_y Y_s^{(t, Y_t)} ds \right) a(Y_t)
           \end{multlined} \\
           &= -H_T \theta^\top (Y_t) + H_t^{(0, y)} \left( \nabla_y H_T^{(t, Y_t)} \right)^\top a(Y_t).
\end{align*}

\begin{proof}[Proof of Theorem~\ref{explicit.portfolio}]
    By Clark's formula (Proposition~\ref{Clark.formula}),
    $\psi$ in Theorem~\ref{Optimal Terminal Wealth} is given by
    \begin{align*}
        \psi_t &= \mathbb{E}_t\left[ D_t \left( H_T I(\hat{\lambda} H_T) \right)^\top  \right] \\
               &= - \mathbb{E}_t\left[ \left( H_T \theta(Y_t) + H_t a^\top(Y_t) \nabla_y H_T^{(t, Y_t)}  \right) \left(I(\hat{\lambda} H_T) + \hat{\lambda} H_T I^\prime(\hat{\lambda}H_T) \right)  \right] \\
               &= - H_t \hat{X}_t \theta(Y_t) - \theta(Y_t) \mathbb{E}_t \left[H_T \cdot \hat{\lambda}H_T I^\prime(\hat{\lambda}H_T) \right] + H_t a^\top(Y_t)   \mathbb{E}_t \left[ \nabla_y H_T^{(t, Y_t)}  \left(I(\hat{\lambda} H_T) + \hat{\lambda} H_T I^\prime(\hat{\lambda}H_T) \right)  \right].
    \end{align*}
    As a result,
    \begin{align*}
        \hat{\pi}_t &= \sigma^\top (Y_t)^{-1} \left( \frac{\psi_t}{H_t} + \hat{X}_t \theta(Y_t) \right) \\
        &=- \sigma^\top (Y_t)^{-1} \theta(Y_t) \frac{1}{H_t} \mathbb{E}_t \left[H_T \cdot \hat{\lambda}H_T I^\prime(\hat{\lambda}H_T) \right] + \sigma^\top (Y_t)^{-1} a^\top(Y_t)   \mathbb{E}_t \left[ \nabla_y H_T^{(t, Y_t)}  \left(I(\hat{\lambda} H_T) + \hat{\lambda} H_T I^\prime(\hat{\lambda}H_T) \right)  \right]. 
    \end{align*}
\end{proof}

\begin{proof}[Proof of Proposition~\ref{feedback.form}]
    Let $\mathcal{X}:[0, \infty) \times (0, \infty) \times \mathbb{R}^m \to (0, \infty)$ be
    \[
        \mathcal{X}(t, z, y) \coloneqq \mathbb{E}^y\left[H_t I(z H_t) \right],
    \]
    where $\mathbb{E}^y[\cdot]$ stands for the expectation when $Y$ starts from $y \in \mathbb{R}^m$ at time 0.
    By the Markov property of $Y$, $ \mathcal{X}(T- t, z, Y_t) =  \mathbb{E}_t \left[ \frac{H_T}{H_t} I \left( z \frac{H_T}{H_t} \right) \right]$. Furthermore, because $H_t$ is $\mathcal{F}_t$-measurable, 
    \begin{align*}
        \mathcal{X}(T - t, \hat{\lambda}H_t, Y_t) = \frac{1}{H_t}   \mathbb{E}_t \left[H_T I \left( \hat{\lambda} H_t \cdot \frac{H_T}{H_t} \right) \right] = \hat{X}_t.
    \end{align*}
    Therefore, if we let $F(t, \cdot, y) \coloneqq \left( \mathcal{X}(t, \cdot, y) \right)^{-1}$, then 
    \begin{equation}\label{F.func}
                F(T - t, \hat{X}_t, Y_t) = \hat{\lambda}H_t
    \end{equation}
    holds. Let $\hat{\pi}^M:[0, \infty) \times (0, \infty) \times \mathbb{R}^m \to \mathbb{R}^n$ be
    \[
        \hat{\pi}^M(t, x, y) \coloneqq - (\sigma^\top(y))^{-1} \theta(y) F(t, x, y) \mathbb{E}^y \left[ \left( H_t \right)^2 I^\prime \left( F(t, x, y) H_t \right) \right].
    \]
    By the Markov property of $Y_t$, 
    \[
    \hat{\pi}^M(T - t, x, Y_t) =  - (\sigma^\top(Y_t))^{-1} \theta(Y_t) F(T - t, x, Y_t) \mathbb{E}_t \left[ \left( \frac{H_T}{H_t} \right)^2 I^\prime \left( F( T - t, x, Y_t) \frac{H_T}{H_t} \right) \right]
    \]
    holds. Moreover, by the identity (\ref{F.func}) and the $\mathcal{F}_t$-measurability of $\hat{X}_t$, 
    \begin{align*}
            \hat{\pi}^M(T - t, \hat{X}_t, Y_t) &=  - (\sigma^\top(Y_t))^{-1} \theta(Y_t) F(T - t, \hat{X}_t, Y_t) \mathbb{E}_t \left[ \left( \frac{H_T}{H_t} \right)^2 I^\prime \left( F( T - t, \hat{X}_t, Y_t) \frac{H_T}{H_t} \right) \right]\\
            &= - (\sigma^\top(Y_t))^{-1} \theta(Y_t) \frac{1}{H_t} \mathbb{E}_t \left[  H_T \cdot \hat{\lambda} H_T I^\prime \left( \hat{\lambda}H_T \right) \right] \\
            &= \hat{\pi}^M_t.
    \end{align*}
    In particular, if we set $T = \tau,\; t = 0$, then
    \[
        \hat{\pi}^M(\tau, x, y) = \hat{\pi}^M_0 = -(\sigma^\top(y))^{-1} \theta(y) \mathbb{E}^y \left[ H_\tau \cdot {\hat{\lambda} H_\tau I^\prime(\hat{\lambda} H_\tau)}\right],
    \]
    where $\hat{\lambda} = \hat{\lambda}(\tau, x, y)$ is defined by an equality $x = \mathbb{E}^y[H_\tau I(\lambda H_\tau)]$.
    As a result, we have proved that the myopic portfolio $\left( \hat{\pi}^M_t \right)_{t \in [0, T]}$ can be represented by a feedback form. By the same argument, we can also prove that the excess hedging demand $\left( \hat{\pi}^H_t \right)_{t \in [0, T]}$ has a feedback form, which completes the proof.
\end{proof}

\subsection{Proofs for Sect.~\ref{Myopic.Turnpike}}
In this subsection, we assume Assumptions \ref{coeff.asp}--\ref{pref.asp} and prove the main results in Sect.~\ref{Myopic.Turnpike}. Lemma~\ref{holder.lem} and Proposition~\ref{d.bound} are useful to estimate the rate of the turnpike theorem for myopic portfolios (Theorem~\ref{myopic.turnpike}) and optimal wealth processes (Theorem~\ref{wealth.turnpike}). Lemmas \ref{f.low.est} and \ref{uni.M.lb} are used to prove the uniform turnpike theorem for optimal portfolio proportions (Theorem~\ref{uni.conv.myopic}).
\begin{lemma}\label{holder.lem}
\[
        \mathbb{E}[H_T^{1+\alpha}]  \leq \mathbb{E}\left[ H_T^{q} \right] ^{\frac{\alpha}{q-1}} \mathbb{E} \left[  H_T \right]^{1 - \frac{\alpha}{q-1}}.
\]
\end{lemma}
\begin{proof}
The case $\alpha = 0$ is trivial. When $\alpha \in (q - 1, 0)$, using Hölder's inequality leads to
\begin{align*}
    \mathbb{E}[H_T^{1+\alpha}] &= \mathbb{E}\left[ H_T^{\frac{q\alpha}{q - 1}} H_T^{\frac{q - 1 - \alpha}{q-1}} \right]\\
                          &\leq \mathbb{E}\left[ H_T^{\frac{q\alpha}{q - 1} \cdot \frac{q-1}{\alpha}} \right] ^{\frac{\alpha}{q-1}} \mathbb{E}\left[  H_T^{\frac{q-1-\alpha}{q-1} \cdot \frac{q-1}{q-1-\alpha}} \right]^\frac{q-1-\alpha}{q-1} \\
                          &= \mathbb{E}\left[ H_T^{q} \right] ^{\frac{\alpha}{q-1}} \mathbb{E} \left[  H_T \right]^{1 - \frac{\alpha}{q-1}}.
\end{align*}
\end{proof}

\begin{proposition}\label{d.bound}
Let $d(z) \coloneqq I_1(z) - I_2(z)$. Then, there exists an $M = M(x, y) \in ( - \infty, x]$, which is independent of $T$, such that
\begin{align*}
        \left| \mathbb{E}[H_T d (\hat{\lambda}^{1, T} H_T)] \right| &\leq  K \left( \mathbb{E}[H_T] +  (\hat{\lambda}^{1, T})^\alpha \mathbb{E}[H_T^{1+\alpha}] \right) \\
        &\leq K \left( \mathbb{E}[H_T] + (x - M)^{\frac{\alpha}{q - 1}} \mathbb{E} [H_T]^{1 - \frac{\alpha}{q-1}} \right), \quad T> 0.
\end{align*}
In particular, 
\[
     \left| \mathbb{E}[H_T d (\hat{\lambda}^{1, T} H_T)] \right| = O \left( \mathbb{E} [H_T]^{1 - \frac{\alpha}{q-1}} \right), \quad (T \nearrow \infty).
\]
\end{proposition}
\begin{proof}
Note that
\begin{equation}\label{ineq.Hd}
    \mathbb{E}[H_T d(\hat{\lambda}^{1, T} H_T) ] = x - \mathbb{E}[H_T I_2 (\hat{\lambda}^{1, T} H_T)] = x - (\hat{\lambda}^{1, T})^{q - 1} \mathbb{E}[H_T^q] \leq x.
\end{equation}
We define $ M = M(x, y) \coloneqq \inf_{T > 0} \mathbb{E}[H_T d (\hat{\lambda}^{1, T} H_T)] \in [-\infty, x]$.
By Assumption~\ref{pref.asp}(ii), (\ref{ineq.Hd}), and Lemma~\ref{holder.lem},
    \begin{align*}
        \left| \mathbb{E}[H_T d (\hat{\lambda}^{1, T} H_T)] \right| &\leq K \left( \mathbb{E}[H_T] +  (\hat{\lambda}^{1, T})^\alpha \mathbb{E}[H_T^{1+\alpha}] \right) \\
        &= K \left\{ \mathbb{E}[H_T] +  \left( \frac{x - \mathbb{E}[H_T d (\hat{\lambda}^{1, T} H_T)]}{\mathbb{E} [H_T^q]} \right)^{\frac{\alpha}{q-1}} \mathbb{E}[H_T^{1+\alpha}] \right\} \\
        &\leq K \left\{ \mathbb{E}[H_T] +  \left( x - \mathbb{E}[H_T d (\hat{\lambda}^{1, T} H_T)]\right)^{\frac{\alpha}{q-1}} \mathbb{E}[H_T]^{1 - \frac{\alpha}{q - 1}} \right\} \\
        &\leq K \left\{ \mathbb{E}[H_T] +  \left( x - M\right)^{\frac{\alpha}{q-1}} \mathbb{E}[H_T]^{1 - \frac{\alpha}{q - 1}} \right\}
    \end{align*}
holds.  Lastly, we prove that $M > - \infty$. Because $\mathbb{E}[H_T] \searrow 0$, there exists some constant $C$ such that $\mathbb{E}[H_T] \leq C \mathbb{E}[H_T]^{1 - \frac{\alpha}{q - 1}}$. As a result, 
\begin{align}\label{H.d.ineq}
        \left| \mathbb{E}[H_T d (\hat{\lambda}^{1, T} H_T)] \right| 
        &\leq K \left\{ C +  \left( x - \mathbb{E}[H_T d (\hat{\lambda}^{1, T} H_T)]\right)^{\frac{\alpha}{q-1}} \right\} \mathbb{E}[H_T]^{1 - \frac{\alpha}{q - 1}}
\end{align}
holds. Dividing both sides of the inequality (\ref{H.d.ineq}) by $|\mathbb{E}[H_T d(\hat{\lambda}^{1, T})]|$, we have 
\begin{equation}\label{contra.inq}
    1 \leq K \left\{ \frac{C}{  \left| \mathbb{E}[H_T d (\hat{\lambda}^{1, T} H_T)] \right| } + \left| \frac{x - \mathbb{E}[H_T d (\hat{\lambda}^{1, T} H_T)]}{\mathbb{E}[H_T d (\hat{\lambda}^{1, T} H_T)]} \right|^{\frac{\alpha}{q - 1}} \cdot \frac{1}{\left| \mathbb{E}[H_T d (\hat{\lambda}^{1, T} H_T)] \right|^{1 - \frac{\alpha}{q - 1}} } \right\} \mathbb{E}[H_T]^{1 - \frac{\alpha}{q - 1}}.    
\end{equation}
If $M = -\infty$, there exists a sequence $(T_n)_{n \geq 1}$ such that $\mathbb{E} \left[H_{T_n} d(\hat{\lambda}^{1, T_n} H_{T_n})\right] \searrow - \infty$ and the inequality (\ref{contra.inq}) leads to $1 \leq 0$. Therefore, $M > -\infty$ and the proof is completed.
\end{proof}

\begin{proof}[Proof of Theorem~\ref{myopic.turnpike}]
Let $J_i (z) \coloneqq z I_i^\prime(z), \; (i = 1, 2)$.
\begin{align*}
    \mathbb{E}\left[H_T \cdot J_1 \left( \hat{\lambda}^{1, T} H_T \right) \right] - \mathbb{E}\left[H_T \cdot J_2 \left( \hat{\lambda}^{2, T} H_T \right) \right] &= 
        \mathbb{E} \left[ H_T \left\{ J_1 \left( \hat{\lambda}^{1, T} H_T \right) - J_2 \left( \hat{\lambda}^{1, T} H_T \right) \right\} \right] \\
        & \qquad + \mathbb{E} \left[ H_T \left\{ J_2 \left( \hat{\lambda}^{1, T} H_T \right) - J_2 \left( \hat{\lambda}^{2, T} H_T \right) \right\} \right] \\
    &\eqqcolon (\text{I}) + (\text{II}).
\end{align*}
Because $J_2 (z) = (q - 1) z^{q-1}$,
\begin{align*}
    (\text{II}) = (q - 1) \left\{ \left(  \hat{\lambda}^{1, T} \right)^{q - 1} - \left( \hat{\lambda}^{2, T} \right)^{q - 1} \right\} \mathbb{E}[H_T^{q}].
\end{align*}
By the definition of $\hat{\lambda}^{i, T}$, 
\begin{align*}
    x &= \mathbb{E}[H_T I_1 (\hat{\lambda}^{1, T}H_T)] = \mathbb{E}[H_T d (\hat{\lambda}^{1, T} H_T)] + (\hat{\lambda}^{1, T})^{q - 1} \mathbb{E}[H_T^q], \\
    x &=   (\hat{\lambda}^{2, T})^{q - 1} \mathbb{E}[H_T^q]
\end{align*}
hold, and it leads to 
\begin{align*}
    (\text{II}) &= (1 - q)  \mathbb{E}[H_T d (\hat{\lambda}^{1, T} H_T)].
\end{align*}
As a result, by Assumption~\ref{pref.asp}(ii) and Proposition~\ref{d.bound}, we obtain
\begin{align*}
     |\hat{\pi}^{1, M} (T, x, y) - \hat{\pi}^{2, M} (T, x, y)|  &\leq K (2 - q) \left| (\sigma^\top(y))^{-1} \theta (y) \right| \left( \mathbb{E} \left[ H_T \right] + (\hat{\lambda}^{1, T})^{\alpha} \mathbb{E} \left[ H_T^{\alpha + 1} \right] \right) \\
     &\leq  K (2 - q) \left| (\sigma^\top(y))^{-1} \theta (y) \right| \left( \mathbb{E}[H_T] + (x - M)^{\frac{\alpha}{q - 1}} \mathbb{E} [H_T]^{1 - \frac{\alpha}{q-1}} \right).
\end{align*}

\end{proof}

 
\begin{lemma}\label{f.low.est}
    For constants $C > 0,\; \gamma \in [0, 1) $, we define a function $f$ by
    \[
        f(x, z) \coloneqq \frac{|z|}{C + (x - z)^\gamma}, \quad x > 0, \; z < -x.
    \]
    Then $f$ satisfies
    \begin{align*}
        z \in [-x, x) &\Rightarrow f(x, z) \geq \frac{|z|}{ C + (2x)^\gamma}, \\
        z \in (-\infty, -x) &\Rightarrow  f(x, z) \geq \frac{|z|^{1 - \gamma}}{C x^{-\gamma} + 2^\gamma}. 
    \end{align*}
\end{lemma}
\begin{proof}
When $z \in [-x, x)$, 
\begin{align*}
    f(x, z) &= \frac{|z|}{C + (x - z)^\gamma} \\
            &\geq \frac{|z|}{C + (x - (-x))^\gamma} \\
            &= \frac{|z|}{ C + (2x)^\gamma}.
\end{align*}
When $z \in (-\infty, -x)$, 
\begin{align*}
    f(x, z) &= \frac{|z|}{C + (x + |z|)^\gamma} \\
    &= \frac{|z|^{1 - \gamma}}{C|z|^{-\gamma} + \left( \frac{x}{|z|} + 1 \right)^\gamma} \\
    &\geq \frac{|z|^{1 - \gamma}}{Cx^{-\gamma} + 2^\gamma}.
\end{align*}
\end{proof}

\begin{lemma}\label{uni.M.lb}
    For any $\epsilon > 0$, 
    \[
        \sup_{x > \epsilon} \left( \frac{x - M(x, y)}{x} \right) < \infty.
    \]
\end{lemma}
\begin{proof}
    Let $G(T, x, y) \coloneqq \mathbb{E}[H_T d(\hat{\lambda}^{1, T} H_T)]$.
    Because
    \[
        \sup_{x > \epsilon} \left( \frac{x - M(x, y)}{x} \right) = 1 - \inf_{x > \epsilon, T > 0} \frac{G(T, x, y)}{x},
    \]
    it suffices to prove 
    \[
    \inf_{x > \epsilon, T > 0} \frac{G(T, x, y)}{x} > - \infty.
    \]
    From (\ref{H.d.ineq}), 
    \[
       f(x, G) = \frac{|G|}{C + (x - G)^{\frac{\alpha}{q - 1}}} \leq K E[H_T]^{1 - \frac{\alpha}{q - 1}},
    \]
    where we define $f$ as in Lemma~\ref{f.low.est} with $\gamma = \frac{\alpha}{q - 1}$. Lemma~\ref{f.low.est} implies that when $\frac{G}{x} < - 1$,
    \[
     \frac{|G|^{1 - \gamma}}{C x^{-\gamma} + 2^\gamma} \leq f(x, G) \leq K E[H_T]^{1 - \gamma}, 
    \]
    which means that
    \begin{align*}
        \frac{G}{x} &\geq -\frac{1}{x} \left\{ K \left( Cx^{-\gamma} + 2^\gamma \right)  \right\}^{\frac{1}{1 - \gamma}} \mathbb{E}[H_T] \\
        &= -\left\{ K \left( Cx^{-1} + 2^\gamma x^{\gamma - 1} \right)  \right\}^{\frac{1}{1 - \gamma}}  \mathbb{E}[H_T] \\
        &> -\infty
    \end{align*} 
    Therefore, we obtain
    \begin{align*}
            \inf_{x > \epsilon, T > 0} \frac{G(T, x, y)}{x} &\geq - \sup_{x > \epsilon} \left\{ K \left( Cx^{-1} + 2^\gamma x^{\gamma - 1} \right)  \right\}^{\frac{1}{1 - \gamma}} \cdot \sup_{T > 0} \mathbb{E}[H_T] \\
            &= -  \left\{ K \left( C\epsilon^{-1} + 2^\gamma \epsilon^{\gamma - 1} \right)  \right\}^{\frac{1}{1 - \gamma}} \cdot \sup_{T > 0} \mathbb{E}[H_T] \\
            &> -\infty.
    \end{align*}
\end{proof}

\begin{proof}[Proof of Theorem~\ref{uni.conv.myopic}]
    From Theorem~\ref{myopic.turnpike} and Lemma~\ref{uni.M.lb}, 
    \begin{align*}
        \sup_{x > \epsilon} \left|  \frac{\hat{\pi}^{1, M} (T, x, y)}{x} - \frac{\hat{\pi}^{2, M} (T, x, y)}{x} \right| &\leq K(2 - q) |\sigma^\top(y)^{-1} \theta(y)| \left( \frac{\mathbb{E}[H_T]}{\epsilon} + \sup_{x > \epsilon}\left( \frac{x - M}{x} \right)^{\frac{\alpha}{q - 1}} \left( \frac{\mathbb{E}[H_T]}{\epsilon} \right)^{1 - \frac{\alpha}{q - 1}} \right) \\
        &= O \left( \mathbb{E}[H_T]^{1 - \frac{\alpha}{q - 1}} \right).
    \end{align*}
\end{proof}

\begin{proof}[Proof of Theorem~\ref{wealth.turnpike}]
    Because $\left(  H_t| \hat{X}^{1, T}_t - \hat{X}^{2, T}_t |\right)_t$ is a submartingale,
    \[
        \mathbb{E}[ H_t| \hat{X}^{1, T}_t - \hat{X}^{2, T}_t |]  \leq  \mathbb{E}[ H_T| \hat{X}^{1, T}_T - \hat{X}^{2, T}_T |] 
    \]
    holds. By the identity $\hat{X}^{i, T}_T = I_i \left( \hat{\lambda}^{i, T} H_T \right)$ and $  \mathbb{E}[H_T d (\hat{\lambda}^{1, T} H_T)] = \left\{ \left(  \hat{\lambda}^{1, T} \right)^{q - 1} - \left( \hat{\lambda}^{2, T} \right)^{q - 1} \right\} \mathbb{E}[H_T^{q}]$, it follows that
    \begin{align*}
        \mathbb{E}\left[ H_T| \hat{X}^{1, T}_T - \hat{X}^{2, T}_T | \right] &\leq \mathbb{E}\left[ H_T \left| I_1 \left( \hat{\lambda}^{1, T} H_T \right) - I_2 \left( \hat{\lambda}^{1, T} H_T \right)  \right| \right] +   \mathbb{E}\left[ H_T \left| I_2 \left( \hat{\lambda}^{1, T} H_T \right) - I_2 \left( \hat{\lambda}^{2, T} H_T \right)  \right| \right] \\
        &= \mathbb{E}\left[ H_T \left| d\left( \hat{\lambda}^{1, T} H_T \right)  \right| \right] + \left| \left( \hat{\lambda}^{1, T} \right)^{q - 1} - \left( \hat{\lambda}^{2, T} \right)^{q - 1} \right| \cdot \mathbb{E}\left[H_T^q\right] \\
        &\leq 2 \mathbb{E}\left[ H_T \left| d\left( \hat{\lambda}^{1, T} H_T \right)  \right| \right]. \\
    \end{align*}
    By combining the above inequalities with Proposition~\ref{d.bound}, we get
    \begin{align*}
        \sup_{t \in [0, T]}\mathbb{E}[ H_t| \hat{X}^{1, T}_t - \hat{X}^{2, T}_t |] &\leq 2 \mathbb{E}\left[ H_T \left| d\left( \hat{\lambda}^{1, T} H_T \right)  \right| \right] \\
        &\leq 2K \left( \mathbb{E}[H_T] +  (\hat{\lambda}^{1, T})^\alpha \mathbb{E}[H_T^{1+\alpha}] \right) \\
        &\leq 2K \left( \mathbb{E}[H_T] + (x - M)^{\frac{\alpha}{q - 1}} \mathbb{E} [H_T]^{1 - \frac{\alpha}{q-1}} \right) \\
        &= O \left( \mathbb{E} [H_T]^{1 - \frac{\alpha}{q-1}} \right), \quad (T \nearrow \infty).
    \end{align*}
\end{proof}

\subsection{Proofs for Sect.~\ref{Excess.Turnpike}}
For excess hedging demands $\hat{\pi}^{i, H}$, the following estimate holds under Assumptions \ref{coeff.asp}--\ref{pref.asp}. 
\begin{proposition}\label{excess.hedging.bdd}
We assume Assumptions \ref{coeff.asp}--\ref{pref.asp}. Let $M_T^\gamma \coloneqq \frac{H_T^\gamma}{\mathbb{E}[H_T^\gamma]},\; d\mathbb{Q}^\gamma_T \coloneqq M_T^\gamma d\mathbb{P}, \; \gamma \in [0, 1],$ and $M \in ( - \infty, x]$ be as in Proposition~\ref{d.bound}. Then
\begin{multline*}
        |\hat{\pi}^{1, H}(T, x, y) - \hat{\pi}^{2, H}(T, x, y) | \leq \left| (\sigma^{-1}(y))^\top a^\top(y)  \right| \left\{ q \left| \mathbb{E}[H_T d (\hat{\lambda}^{1, T} H_T)] \right| \cdot \mathbb{E}^{\mathbb{Q}^q_T}[|L_T|] \right. \\ \left. + 2K \left( \mathbb{E}[H_T] \cdot \mathbb{E}^{\mathbb{Q}^1_T}[|L_T|] + (x - M)^{\frac{\alpha}{q - 1}} \cdot \mathbb{E}[H_T]^{1 - \frac{\alpha}{q - 1}} \cdot \mathbb{E}^{\mathbb{Q}^{1+\alpha}_T}[|L_T|]\right) \right\},
\end{multline*}
where 
\begin{equation}\label{L}
    L_T \coloneqq \int_0^T (Dr(Y_u) \nabla_y Y_u )^\top du + \int_0^T \left( D\theta (Y_u) \nabla_y Y_u \right)^\top  dW_u + \int_0^T \left( D\theta (Y_u) \nabla_y Y_u \right)^\top \theta(Y_u) du.
\end{equation}

\end{proposition}
\begin{proof}
Recall that the excess hedging demand is given by 
\begin{align*}
    \hat{\pi}^{i, H}(T, x, y) &= (\sigma^\top(y))^{-1}  a^\top(y) \mathbb{E} \left[ \nabla_y H_T  F_i\left( \hat{\lambda}^{i, T}H_T \right) \right],
\end{align*}
where $F_i(z) \coloneqq I_i(z) + z I_i^\prime(z)$ and $\nabla_y H_T = - H_T L_T$.

\begin{align*}
    \mathbb{E} \left[ \nabla_y H_T \left( F_1(\hat{\lambda}^{1, T} H_T) - F_2( \hat{\lambda}^{2,T} H_T ) \right) \right] &=\mathbb{E} \left[ \nabla_y H_T \left( F_1(\hat{\lambda}^{1, T} H_T) - F_2( \hat{\lambda}^{1,T} H_T ) \right) \right] \\ & \qquad \qquad \qquad + \mathbb{E} \left[ \nabla_y H_T \left( F_2(\hat{\lambda}^{1, T} H_T) - F_2( \hat{\lambda}^{2,T} H_T ) \right) \right] \\
    & \eqqcolon (\text{I}) + (\text{II}).
\end{align*} 
First, we evaluate $(\text{I})$. Assumption~\ref{pref.asp}(ii) implies $|F_1(z) - F_2(z)| \leq 2K(1 + z^\alpha)$ and Proposition~\ref{d.bound} implies $(\hat{\lambda}^{1, T})^\alpha \mathbb{E}[H_T^{1+\alpha}] \leq (x - M)^{\frac{\alpha}{q - 1}} \mathbb{E} [H_T]^{1 - \frac{\alpha}{q-1}}$. Using these inequalities and myopic probability measures $d\mathbb{Q}^{\gamma}_T \coloneqq \frac{H_T}{\mathbb{E}[H_T]} d\mathbb{P}$, it follows that 
\begin{equation}\label{I.eval}
    \begin{aligned}
                |(\text{I})| &\leq \mathbb{E} \left[ H_T |L_T| \left| F_1(\hat{\lambda}^{1, T} H_T) - F_2( \hat{\lambda}^{1,T} H_T ) \right| \right]  \\
        &\leq 2K\left( \mathbb{E}[|L_T| H_T] + \left( \hat{\lambda}^{1, T} \right)^{\alpha} \mathbb{E}[|L_T| H_T^{\alpha + 1}] \right) \\
        &\leq 2K \left( \mathbb{E}[H_T] \mathbb{E}^{\mathbb{Q}^1_T}[|L_T|] + (x - M)^{\frac{\alpha}{q - 1}} \mathbb{E}[H_T]^{1 - \frac{\alpha}{q - 1}} \mathbb{E}^{\mathbb{Q}^{1+\alpha}_T}[|L_T|]\right).
    \end{aligned}
\end{equation}
Next, we evaluate $(\text{II})$. Because $F_2(z) = q z^{q-1}$,
\[
    (\text{II}) = q \left\{ \left( \hat{\lambda}^{1, T} \right)^{q-1} - \left( \hat{\lambda}^{2, T} \right)^{q-1} \right\} \mathbb{E}\left[ \nabla_y H_T H_T^{q-1} \right].
\]
Because $\hat{\lambda}^{i, T}\; (i = 1, 2)$ are determined by the identity
\begin{align*}
    x &= \mathbb{E}\left[H_T I_1(\hat{\lambda}^{1,T} H_T)\right] =  \left( \hat{\lambda}^{1, T} \right)^{q-1} \mathbb{E}\left[H_T^{q}\right] + \mathbb{E}\left[H_T d(\hat{\lambda}^{1,T} H_T)\right],\\
    x &= \mathbb{E}\left[H_T I_2(\hat{\lambda}^{2,T} H_T)\right] = \left( \hat{\lambda}^{2, T} \right)^{q-1} \mathbb{E}\left[H_T^{q}\right],
\end{align*}
it follows that
\[
    \left( \hat{\lambda}^{1, T} \right)^{q-1} - \left( \hat{\lambda}^{2, T} \right)^{q-1} = - \frac{\mathbb{E}\left[H_T d(\hat{\lambda}^{1,T} H_T)\right]}{\mathbb{E}\left[H_T^{q}\right]},
\]
which leads to
\begin{equation}\label{II.eval}
    \begin{aligned}
    (\text{II}) &=  -q\mathbb{E}\left[H_T d(\hat{\lambda}^{1,T} H_T)\right] \frac{\mathbb{E}\left[ \nabla_y H_T H_T^{q-1} \right]}{\mathbb{E}\left[H_T^{q}\right]}\\
    &= q\mathbb{E}\left[H_T d(\hat{\lambda}^{1,T} H_T)\right] \mathbb{E}^{\mathbb{Q}^{q}_T}\left[ L_T\right].
    \end{aligned}
\end{equation}
Therefore, the statement follows from (\ref{I.eval}) and (\ref{II.eval}).
\end{proof}
To estimate the convergence rate of the turnpike theorem for excess hedging demands in stochastic factor models (\ref{gen.stoch.fac.model}), the above proposition implies that it suffices to derive $O \left(\mathbb{E}^{\mathbb{Q}_T^\gamma} [|L_T|] \right)$ for $\gamma \in [0, 1]$, where $\mathbb{Q}^\gamma_T$ denotes the myopic probability measures used in \cite{GKRX} and $L_T$ is given by (\ref{L}).
Because the martingale density processes of myopic probabilities $\mathbb{Q}^\gamma_T$ can be computed by the optimal portfolios for CRRA investors, which can be represented by solutions to semilinear PDEs, the estimation of $O \left(\mathbb{E}^{\mathbb{Q}_T^\gamma} [|L_T|] \right)$ requires the asymptotic behavior of the solutions to semilinear PDEs. Here, we restrict our models to the quadratic term structure model given by (\ref{QTSM1}) and (\ref{QTSM2}) and use the asymptotic properties of the solutions to the Riccati differential equations. Further research for general stochastic factor models (\ref{gen.stoch.fac.model}) is postponed to future work.
\begin{remark}
When $\gamma = 0$, $\mathbb{Q}^\gamma_T = \mathbb{P}$. When $\gamma = 1$, $d\mathbb{Q}^\gamma_T = \frac{H_T}{\mathbb{E}[H_T]} d\mathbb{P}$ is a $T$-forward measure under which the price process of a zero-coupon bond is chosen as numéraire. When $\gamma \in (0, 1)$, as Proposition~\ref{CRRA.opt.inv.prob} will say, $d\mathbb {Q}^\gamma_T = H_T \cdot \hat{X}^T_T d\mathbb{P}$, where $\hat{X}^T$ is the optimal wealth process for a CRRA investor with a utility function $x \mapsto \frac{x^p}{p},\; p = \frac{\gamma}{\gamma - 1}$ and initial unit wealth. Therefore, in the case of $\gamma \in (0, 1)$, we choose $\hat{X}^T$ as numéraire under the myopic probability $\mathbb{Q}^\gamma_T$.
\end{remark}
\subsubsection{Martingale density processes of the myopic probabilities}\label{heg.bond}
In Sects.~\ref{heg.bond} and \ref{proof.excess.turnpike}, we consider the quadratic term structure model given by (\ref{QTSM1}) and (\ref{QTSM2}) and assume Assumptions \ref{pref.asp} and \ref{stable.asp}.
In this subsection, we first consider the pricing and hedging problem for a $T$-bond (Proposition~\ref{myopic.density.gamma1}) and utility maximization problems for CRRA investors (Proposition~\ref{CRRA.opt.inv.prob}).
Using these results, we compute martingale density processes of myopic probabilities $\mathbb{Q}^\gamma_T$ (Proposition~\ref{myopic.prob.density}).  
\begin{proposition}\label{myopic.density.gamma1} 
The price of a $T$-bond at time $t$ is given by
\begin{equation}\label{price.zc}
       F(t, y) = \mathbb{E}^\mathbb{Q} \left[ \left. \exp \left( - \int_t^T r(Y_v) dv \right)\right| Y_t = y \right] = \exp \left( -\alpha(t;T) - \beta(t;T)^\top y - \frac{1}{2} y^\top C(t;T) y \right),
\end{equation}
where $\alpha(\cdot; T):[0, T] \to \mathbb{R}, \;\beta(\cdot;T):[0, T] \to \mathbb{R}^m,$ and $ C(\cdot; T):[0, T] \to \mathbb{S}^m_+$ solve the following system of ordinary differential equations:
\begin{align}
    &\begin{aligned}\label{C.Riccati}
        &\dot{C}(t) - C(t)^\top \Lambda \Lambda^\top C(t) + \Tilde{B}^\top C(t) + C(t) \Tilde{B} + R_2 = 0,\\
        & C(T) = 0, 
    \end{aligned}\\
    &\begin{aligned}\label{beta.Riccati}
        &\dot{\beta}(t) + (\Tilde{B} - \Lambda \Lambda^\top C(t)) \beta(t) + C^\top(t) \Tilde{b} + r_1 = 0,\\
        &\beta(T) = 0, \\
    \end{aligned}\\
    &\begin{aligned}\label{alpha.Riccati}
        &\dot{\alpha}(t) + \frac{1}{2}\mathrm{Tr}(\Lambda \Lambda^\top C(t)) - \frac{1}{2}\beta^\top(t) \Lambda \Lambda^\top \beta(t) + \Tilde{b}\beta(t) + r_0 = 0,\\
        &\alpha(T) = 0,
    \end{aligned}
\end{align}
where  $\Tilde{B} \coloneqq B - \Lambda A, \; \Tilde{b} \coloneqq b - \Lambda a$.
Furthermore, the portfolio proportion process $\hat{\pi}$ that hedges a $T$-bond and the corresponding wealth process $X^{x, \hat{\pi}}$ with initial wealth $x = F(0, y)$ are given by 
\begin{align*}
    \hat{\pi}_t &= - (\Lambda \Sigma^{-1})^\top \left( \beta(t;T) + C(t;T) Y_t \right), \\
    X^{x, \hat{\pi}}_t &= F(t, Y_t).
\end{align*}

\end{proposition}
\begin{proof} 
This proposition follows from well-known arguments in option pricing theory, which for convenience we include in Appendix~\ref{opt.pricing}. By Theorem~\ref{comp.stoc.hedg}, it suffices to check that the pricing PDE (\ref{pricing.PDE}) for a $T$-bond has a unique solution $F$ and $(H_t X^{x, \pi}_t)$ is $\mathbb{P}$-martingale, where $x$ and $\pi$ are given in terms of the solution $F$.
First, the pricing PDE is given by
\begin{equation}\label{pricing.PDE}
\begin{aligned}
        \partial_t F + D_yF^\top \left( b - \Lambda a + (B - \Lambda A) y \right) + \frac{1}{2}\mathrm{Tr} \left[ \Lambda \Lambda^\top D_{yy}F \right] - r(y) F &= 0,\\
        F(T, y) &= 1.
\end{aligned}
\end{equation}
If $F$ has the form 
\[
    F(t, y) = \exp \left( -\alpha(t;T) - \beta(t;T)^\top y - \frac{1}{2} y^\top C(t;T) y \right), 
\]
then  $\alpha(\cdot; T):[0, T] \to \mathbb{R}, \;\beta(\cdot;T):[0, T] \to \mathbb{R}^m,$ and $ C(\cdot; T):[0, T] \to \mathbb{S}^m_+$ are solutions to the ODEs (\ref{C.Riccati}), (\ref{beta.Riccati}), and (\ref{alpha.Riccati}).
By Theorem~\ref{exist.uniq.RDE}, the Riccati equation (\ref{C.Riccati}) has a unique solution, and thus the linear ODEs (\ref{beta.Riccati}) and (\ref{alpha.Riccati}) have unique solutions. Therefore, Theorem~\ref{comp.stoc.hedg}(i) implies that the replicating cost $x$, the hedging portfolio proportion process $ \hat{\pi} = (\hat{\pi_t})_{t \in [0, T]}$, and the corresponding wealth process $X^{x, \hat{\pi}}$ for a $T$-bond are given by
\begin{align*}
x = F(0, y), \quad    \hat{\pi}_t = - (\Lambda \Sigma^{-1})^\top \left( \beta(t;T) + C(t;T) Y_t \right), \quad  X^{x, \hat{\pi}}_t  = F(t, Y_t). 
\end{align*}
By Ito's formula,
\begin{align*}
     H_t X^{x, \hat{\pi}}_t &= x + \int_0^t H_s X^{x, \hat{\pi}}_s \left( \Sigma^\top \hat{\pi}_s - \theta(Y_s) \right)^\top dW_s \\
     &= x + \int_0^t H_s X^{x, \hat{\pi}}_s  \left\{ -( \Lambda^\top C(s;T) + A) Y_s - \Lambda^\top \beta(s;T) - a \right\}^\top dW_s,
\end{align*}
which means that
\begin{align*}
    H_t X^{x, \hat{\pi}}_t = x \cdot \mathcal{E} \left( \int_0^{\cdot} \left\{ -( \Lambda^\top C(s;T) + A) Y_s - \Lambda^\top \beta(s;T) - a \right\}^\top dW_s \right)_t.
\end{align*}
By the same argument as in \cite[Section~6.2]{LS}, $(H_t X^{x, \hat{\pi}}_t)$ is a $\mathbb{P}$-martingale and thus Theorem~\ref{comp.stoc.hedg}(ii) implies (\ref{price.zc}).
\end{proof}

Next, we recall results on a utility maximization problem for CRRA utility, $x \mapsto \frac{x^p}{p}, \; p < 0$.
    \begin{equation}\label{CRRA.value}
        \begin{aligned}
            V(t, x, y) &\coloneqq \sup_{\pi \in \mathcal{A}(x)} \mathbb{E} \left[ \left.\frac{(X^\pi_T)^p}{p} \right| X^\pi_t = x, Y_t = y  \right],& (t, x, y) &\in [0, T) \times (0, \infty) \times \mathbb{R}^m. \\
            V(T, x, y) &\coloneqq \frac{x^p}{p},&     (x, y) &\in (0, \infty) \times \mathbb{R}^m.
        \end{aligned}
    \end{equation}

\begin{proposition}\label{CRRA.opt.inv.prob}
The value function $V$ is given by 
\[
    V(t, x, y) = \frac{x^p}{p} \exp \left( -\frac{1}{2} y^\top P(t;T) y -q(t;T)^\top y - k(t;T) \right), \quad (t, x, y) \in [0, T] \times (0, \infty) \times \mathbb{R}^m,
\]
where $P(\cdot;T):[0, T] \to \mathbb{S}^m, q(\cdot;T):[0, T] \to \mathbb{R}^m, k(\cdot;T):[0, T] \to \mathbb{R}$ satisfy the following system of ODEs:
    \begin{align}
        &\begin{aligned}\label{P.Riccati}
            &\dot{P}(t) - P(t)^\top K_0 P(t) + K_1^\top P(t) + P(t) K_1 + \frac{p}{p-1} A^\top A - p R_2 = 0,\\
            &P(T) = 0,
        \end{aligned} \\
        &\begin{aligned}\label{q.Riccati}
            &\dot{q}(t) + \left\{ K_1 - K_0 P(t)  \right\}^\top q(t) + P(t) b + \frac{p}{p-1} \left( A -  \Lambda^\top P(t)  \right)^{\top} a  -p r_1= 0,\\
            &q(T) = 0,
        \end{aligned} \\
        &\begin{aligned}\label{k.Riccati}
            &\begin{multlined}
                \dot{k}(t) + \frac{1}{2} \mathrm{Tr}(\Lambda \Lambda^\top P(t)) + \frac{1}{2(p-1)} q(t)^\top \Lambda \Lambda^\top q(t) + \left\{ b - \frac{p}{p-1}\Lambda a   \right\}^\top q(t)\\
                + \frac{p}{2(p-1)} ||a||^2 -p r_0 = 0,
            \end{multlined}\\
            &k(T) =0,
        \end{aligned}
    \end{align}
    where $K_0 \coloneqq \frac{1}{1-p}\Lambda\Lambda^\top,\; K_1 \coloneqq B - \frac{p}{p-1}\Lambda A$.
    Moreover, the optimal portfolio proportion process $\hat{\pi}_t$ and the optimal terminal wealth $X^{x, \hat{\pi}}_T$  for the problem (\ref{CRRA.value}) are given by
    \begin{align}
        \hat{\pi}_t &= \frac{1}{1 - p} (\Sigma^\top)^{-1} \theta(Y_t) - \frac{1}{1 - p} (\Sigma^\top)^{-1}\Lambda^\top\left( P(t;T)Y_t + q(t;T) \right), \label{cand.opt.portfolio} \\
        X^{x, \hat{\pi}}_T &= x \frac{H_T^{q-1}}{\mathbb{E}[H_T^q]}. \label{opt.terminal.wealth}
    \end{align}
\end{proposition}
\begin{proof}
By Theorem~\ref{exist.uniq.RDE}, the Riccati equation (\ref{P.Riccati}) and linear ODEs (\ref{q.Riccati}) and (\ref{k.Riccati}) have unique solutions. Therefore 
\[
    V(t, x, y) \coloneqq \frac{x^p}{p} \exp \left( -\frac{1}{2} y^\top P(t;T) y -q(t;T)^\top y - k(t;T) \right), \quad (t, x, y) \in [0, T] \times (0, \infty) \times \mathbb{R}^m
\]
is the solution to the HJB equation,
\begin{align*}
    & V_t + xr(y) V_x + (b +By)^\top D_y V + \frac{1}{2} \mathrm{Tr}(\Lambda \Lambda^\top D^2_{yy}V)
          - \frac{1}{2V_{xx}} \left|\theta(y) V_x + \Lambda^\top D_y V_x\right|^2 = 0, \\
    &V(T, x, y) = \frac{x^p}{p},
\end{align*}
and the candidate for the optimal portfolio proportion process $\hat{\pi}$ is given by
\[
        \hat{\pi}_t = \frac{1}{1 - p} (\Sigma^\top)^{-1} \theta(Y_t) - \frac{1}{1 - p} (\Sigma^\top)^{-1}\Lambda^\top\left( P(t;T)Y_t + q(t;T) \right). 
\]
Here, we do not use standard verification arguments. Instead, we directly show that the terminal wealth obtained by the candidate $\hat{\pi}$ matches that of the martingale methods, that is, we establish the identity (\ref{opt.terminal.wealth}). To do so, Theorem~\ref{control.equiv.dual} implies that it suffices to check that $\left( H_t X^{x, \hat{\pi}}_t \right)_{t \in [0, T]}$ is a $\mathbb{P}$-martingale. Because
\begin{equation}\label{HX.rep1}    
\begin{aligned}
    H_t X^{x, \hat{\pi}}_t &=x \cdot \mathcal{E} \left( \int_0^\cdot \left( \Sigma^\top \hat{\pi}_s - \theta(Y_s) \right)^\top dW_s \right) \\
    &= x  \cdot \mathcal{E} \left( \int_0^\cdot \left( - \frac{p}{p-1} \theta(Y_s) + \frac{1}{p - 1} \Lambda^\top \left( P(s;T)Y_s + q(s;T) \right) \right)^\top dW_s \right)  \\
    &=x \cdot \mathcal{E} \left( \int_0^{\cdot} \left\{ \left( - \frac{p}{p-1} A + \frac{1}{p-1} \Lambda^\top P(s;T) \right) Y_s -\frac{p}{p-1} a + \frac{1}{p-1} \Lambda^\top q(s;T) \right\}^\top dW_s \right)_t,
\end{aligned}
\end{equation}
the same argument as in \cite[Section~6.2]{LS} implies that $\left( H_t X^{x, \hat{\pi}}_t \right)_{t \in [0, T]}$ is a $\mathbb{P}$-martingale, which completes the proof.
\end{proof}
\begin{remark}
        If $p < 0$, $P(t;T):[0, T] \to \mathbb{S}^m_+$ always exists. If $0 < p < 1$, the solution $P$ may blow up at a finite time. See \cite{HK} for details.
\end{remark}

\begin{proposition}\label{myopic.prob.density}
Let $\gamma \in [0, 1]$. We denote by $P^\gamma, q^\gamma$ the solutions to the system of ODEs (\ref{P.Riccati}), (\ref{q.Riccati}) for $p = \frac{\gamma}{\gamma - 1}, \; \gamma \in(0, 1)$. Let $C^{\gamma}(\cdot; T):[0, T] \to \mathbb{S}^m_+,\; \beta^\gamma(\cdot;T):[0, T] \to \mathbb{R}^m$ be given by 
\begin{align*}
    C^\gamma (t;T) &\coloneqq \begin{cases}
        0, & t \in [0, T],\; \gamma  = 0,\\
        (1 - \gamma)P^\gamma(t;T), & t \in [0, T],\; \gamma \in (0, 1), \\
        C(t;T), & t \in [0, T],\; \gamma = 1,
    \end{cases} \\
    \beta^\gamma (t;T) &\coloneqq \begin{cases}
            0, & t \in [0, T],\; \gamma  = 0,\\
        (1 - \gamma)q^\gamma(t;T), & t \in [0, T],\; \gamma \in (0, 1), \\
        \beta(t;T), & t \in [0, T],\; \gamma = 1.
    \end{cases}
\end{align*}
Then, the martingale density processes of myopic probability measures $\mathbb{Q}^\gamma_T$ are given by
\[
        M^\gamma_T = \mathcal{E} \left( \int_0^\cdot \left[ - \left\{ \gamma A + \Lambda^\top C^\gamma(t;T) \right\} Y_t -\left\{ \gamma a + \Lambda^\top \beta^\gamma(t;T) \right\} \right]^\top dW_t \right)_T.
\]
Hence, $W^{\mathbb{Q}^\gamma_T} = \left( W^{\mathbb{Q}^\gamma_T}_t \right)_{t \in [0, T]}$, given by 
\begin{equation}\label{WQgamma}
    W^{\mathbb{Q}^\gamma_T}_t \coloneqq W_t + \left\{ \left( \gamma A + \Lambda^\top C^\gamma(t;T) \right) Y_t + \left( \gamma a + \Lambda^\top \beta^\gamma(t;T)  \right) \right\} dt,
\end{equation}
is an $n$-dimensional Brownian motion under the myopic probability $\mathbb{Q}^{\gamma}_T$.
Moreover, $Y = (Y_t)_{t \in [0, T]}$ satisfies
\begin{equation}\label{SDE.Y.myopic}
    dY_t = \left\{ b - \gamma \Lambda  a - \Lambda \Lambda^\top \beta^\gamma(t;T)  + \left( B - \gamma \Lambda A - \Lambda \Lambda^\top C^\gamma(t;T) \right) Y_t   \right\} dt + \Lambda dW^{\mathbb{Q}^\gamma_T}_t, \quad Y_0 = y.   
\end{equation}
\end{proposition}
\begin{proof}
For $\gamma = 0$, $\mathbb{Q}^\gamma_T = \mathbb{P}$ and the statements follow immediately. For $\gamma = 1$, by Proposition~\ref{myopic.density.gamma1} and its proof, the state price density process $H$ admits the following stochastic exponential representation:
\[
    \frac{H_T}{\mathbb{E}[H_T]}=  \mathcal{E} \left( \int_0^{\cdot} \left\{ -( \Lambda^\top C(s;T) + A) Y_s - \Lambda^\top \beta(s;T) - a \right\}^\top dW_s \right)_T.
\]
From Girsanov's theorem, $W^{\mathbb{Q}^\gamma_T}$ given by (\ref{WQgamma}) is a Brownian motion under $\mathbb{Q}^\gamma_T$ and (\ref{SDE.Y.myopic}) follows. Therefore, the statements hold for $\gamma = 1$.
Next, we consider the case of $\gamma \in(0, 1)$. By Proposition~\ref{CRRA.opt.inv.prob} and (\ref{HX.rep1}) for $p = \frac{\gamma}{\gamma - 1}$, we have 
\begin{align*}
    M^\gamma_T &= H_T \cdot \frac{H_T^{\gamma - 1}}{\mathbb{E}[H_T^\gamma]} \\
    &= \frac{1}{x} \cdot H_T X^{x, \hat{\pi}}_T \\
    &= \mathcal{E} \left( \int_0^{\cdot} \left[ \left\{ - \gamma A + (\gamma - 1) \Lambda^\top P^\gamma(s;T) \right\} Y_s - \left\{ \gamma a + (1 - \gamma) \Lambda^\top q^\gamma(s;T) \right\} \right]^\top dW_s \right)_T \\
    &= \mathcal{E} \left( \int_0^\cdot \left[ - \left\{ \gamma A + \Lambda^\top C^\gamma(t;T) \right\} Y_t -\left\{ \gamma a + \Lambda^\top \beta^\gamma(t;T) \right\} \right]^\top dW_t \right)_T.
\end{align*}
By Girsanov's theorem again, the statements follow for $\gamma \in (0, 1)$.
\end{proof}
\begin{remark}\label{C.beta.gamma.ODE}
By (\ref{P.Riccati}) and (\ref{q.Riccati}), for $\gamma \in (0, 1)$, $C^\gamma(\cdot;T)$ and $\beta^\gamma(\cdot;T)$ satisfy
    \begin{align}
        &\begin{aligned}\label{C.gamma.Riccati}
            &\dot{C}^\gamma(t) - C^\gamma(t)^\top \Lambda \Lambda^\top C^\gamma(t) + K_1^\top C^\gamma(t) + C^\gamma(t) K_1 + \gamma (1 - \gamma) A^\top A +\gamma R_2 = 0,\\
            &C^\gamma(T) = 0,
        \end{aligned} \\
        &\begin{aligned}\label{beta.gamma.Riccati}
            &\dot{\beta}^\gamma(t) + \left\{ K_1 - \Lambda \Lambda^\top C^\gamma(t)  \right\}^\top \beta^\gamma(t) + C^\gamma(t) \{b - \gamma \Lambda a \} + \gamma  (1 - \gamma) A^{\top} a  + \gamma r_1= 0,\\
            &\beta^\gamma(T) = 0,
        \end{aligned} 
    \end{align}
    where $K_1 \coloneqq B - \gamma\Lambda A$. 
Compared with (\ref{C.Riccati}) and (\ref{beta.Riccati}), $C^\gamma(\cdot;T)$ and $\beta^\gamma(\cdot;T)$ seem to converge to $C^1(\cdot;T) \coloneqq C(\cdot;T)$ and $\beta^1(\cdot;T) \coloneqq \beta(\cdot;T)$ in some sense as $\gamma \nearrow 1$. However, we do not prove it because it is not needed for our main results.
\end{remark}

\subsubsection{Proofs of main results in Sect.~\ref{Excess.Turnpike}}\label{proof.excess.turnpike}
In this subsection, we first consider the asymptotic behavior of $C^\gamma(t; T)$ and $\beta^\gamma(t; T)$ when $T \nearrow \infty$ (Proposition~\ref{asmp.C.beta}), which affects the asymptotic moments of $Y_t$ under myopic probabilities $\mathbb{Q^\gamma_T}$. Using this proposition, we show 
    \[
       \sup_{t \in [0, T]} \mathbb{E}^{\mathbb{Q}^\gamma_T}[|Y_t|^2] = O(1), \quad \quad \mathbb{E}^{\mathbb{Q}^{\gamma}_T}  [|L_T|] = O(1), \quad (T \nearrow \infty)
    \]
in Propositions \ref{Y.myopic.upb} and \ref{LT.bdd}. 
Combining these estimates with Proposition~\ref{excess.hedging.bdd}, we can derive the rate of the turnpike theorem for excess hedging demands (Theorem~\ref{excess.turnpike}). We can also prove the uniform turnpike theorem (Theorem~\ref{uni.excess.turnpike}) by the same arguments as in the case of myopic portfolios. 

\begin{proposition}\label{asmp.C.beta}
For $\gamma \in \{q, \alpha + 1, 1\}$, there exist $C^\gamma_\infty \in \mathbb{S}^m_+$ and $\beta^\gamma_\infty \in \mathbb{R}^m$ such that 
    \begin{align*}
        \lim_{T \nearrow \infty} C^\gamma(t;T) &= C^\gamma_\infty, \\
        \lim_{T \nearrow \infty} \beta^\gamma (t;T) &= \beta^\gamma_\infty
    \end{align*}
    for any $t$. In addition, $B - \gamma \Lambda A - \Lambda \Lambda^\top C^\gamma_\infty$ is stable.
\end{proposition}
\begin{proof}
    Because the statements for $\gamma = 0$ are obvious and those for $\gamma = 1$ are along the lines for the case of $\gamma \in (0, 1)$, we consider only the case of $\gamma \in(0, 1)$. As Remark~\ref{C.beta.gamma.ODE} shows,  $C^\gamma(\cdot; T)$ and $\beta^\gamma(\cdot;T)$ satisfy
        \begin{align}
        &\begin{aligned}
            &\dot{C}^\gamma(t) - C^\gamma(t)^\top \Lambda \Lambda^\top C^\gamma(t) + K_1^\top C^\gamma(t) + C^\gamma(t) K_1 + \gamma (1 - \gamma) A^\top A +\gamma R_2 = 0,\\
            &C^\gamma(T) = 0,
        \end{aligned} \\
        &\begin{aligned}
            &\dot{\beta}^\gamma(t) + \left\{ K_1 - \Lambda \Lambda^\top C^\gamma(t)  \right\}^\top \beta^\gamma(t) + C^\gamma(t) \{b - \gamma \Lambda a \} + \gamma  (1 - \gamma) A^{\top} a  + \gamma r_1= 0,\\
            &\beta^\gamma(T) = 0,
        \end{aligned} 
    \end{align}
    where $K_1 \coloneqq B - \gamma\Lambda A$. By Theorems \ref{Asmp.RDE} and \ref{Asmp.LODE}, we have to check that $\left(K_1,  \Lambda \right)$ is stabilizable and $\left( C, K_1 \right)$ is detectable, where $C \in \mathbb{R}^{n \times m}$ satisfies $C^\top C =  \gamma (1 - \gamma) A^\top A + \gamma R_2$. The stabilizability of $\left(K_1, \Lambda \right)$ can be seen from the fact that 
    \begin{align*}
        &K_1 + \Lambda L = B
    \end{align*}
    is stable when setting $L \coloneqq \gamma A$. To see the detectability of $\left( C, K_1 \right)$, we consider two cases of (iii) in Assumption~\ref{LQMarket.asp}. If $R_2 = 0$, then $C = \sqrt{\gamma(1 - \gamma)} A$ and setting $F = \sqrt{\frac{\gamma}{1 - \gamma}} \Lambda$ implies that
    \[
        FC + K_1 = B,
    \]
    which means that $\left( C, K_1 \right)$ is detectable. If $\left( \gamma (1 - \gamma) A^\top A +\gamma R_2 \right)$ is positive definite, then $C \coloneqq \left( \gamma (1 - \gamma) A^\top A +\gamma R_2 \right)^{\frac{1}{2}}$ is positive definite and thus $\left( C, K_1 \right)$ is detectable. As a result, by Theorems \ref{Asmp.RDE} and \ref{Asmp.LODE}, $C^\gamma_\infty \coloneqq \lim_{T \nearrow \infty} C^\gamma(t;T)$ and $\beta^\gamma_\infty \coloneqq \lim_{T \nearrow \infty} \beta^\gamma(t;T)$ exist for any $t \geq 0$ and satisfy
    \begin{align*}
            - C^\gamma_\infty \Lambda \Lambda^\top C^\gamma_\infty + K_1^\top C^\gamma_\infty + C^\gamma_\infty K_1 + \gamma (1 - \gamma) A^\top A +\gamma R_2 &= 0,\\
            \left\{ K_1 - \Lambda \Lambda^\top C^\gamma_\infty  \right\}^\top \beta^\gamma_\infty + C^\gamma_\infty \{b - \gamma \Lambda a \} + \gamma  (1 - \gamma) A^{\top} a  + \gamma r_1 &= 0,
    \end{align*}
    and 
    \begin{align*}
        K_1 - \Lambda \Lambda^\top C_\infty^\gamma = B - \gamma \Lambda A - \Lambda \Lambda^\top C^\gamma_\infty
    \end{align*}
    is stable.
\end{proof}

\begin{proposition}\label{Y.myopic.upb}
For $\gamma \in \{q, \alpha + 1, 1\}$,
\[
   \sup_{t \in [0, T]} \mathbb{E}^{\mathbb{Q}^\gamma_T} [|Y_t|^2] = O(1), \quad (T \nearrow \infty).
\]
\end{proposition}
\begin{proof}
We define $\Tilde{\beta}^\gamma, \Tilde{C}^\gamma$ by 
\begin{align*}
    \Tilde{\beta}^\gamma (t; T) &\coloneqq  b - \gamma \Lambda  a - \Lambda \Lambda^\top \beta^\gamma(t;T) , \\
    \Tilde{C}^{\gamma} (t; T) &\coloneqq B - \gamma \Lambda A - \Lambda \Lambda^\top C^\gamma(t;T),
\end{align*}
and the SDE (\ref{SDE.Y.myopic}) becomes
\[
    dY_t = \left( \Tilde{\beta}^\gamma(t;T) + \Tilde{C}^\gamma(t;T) Y_t \right) dt + \Lambda dW^{\mathbb{Q}^\gamma_T}_t, \quad Y_0 = y.
\]
Because $\Tilde{C}^\gamma_{\infty} \coloneqq \lim_{T \nearrow \infty} \Tilde{C}^\gamma(t;T) =B - \gamma \Lambda A - \Lambda \Lambda^\top C^\gamma_\infty$ is stable, there exists $U \in \mathbb{S}^m_{++}$ such that
    \[
        (\Tilde{C}^\gamma_{\infty})^\top U + U \Tilde{C}^\gamma_{\infty} < 0. 
    \]
    Therefore, there exist $T_1, \epsilon > 0$ such that
    \[
      T - t \geq T_1 \Rightarrow \Tilde{C}^{\gamma} (t; T)^\top U + U\Tilde{C}^{\gamma} (t; T) < - \epsilon I_m,
    \]
    where $I_m$ is an $m \times m$ identity matrix.
    For $0 \leq s \leq t \leq T$,
    \begin{align*}
        \mathbb{E}^{\mathbb{Q}^{\gamma}_T}  \left[\langle UY_t, Y_t \rangle \right] &= \mathbb{E}^{\mathbb{Q}^{\gamma}_T}  \left[ \langle UY_s, Y_s \rangle \right] + \int_s^t \mathbb{E}^{\mathbb{Q}^{\gamma}_T}  \left[ \left\langle  \left(  \Tilde{C}^{\gamma} (u; T)^\top U + U\Tilde{C}^{\gamma} (u; T) \right) Y_u, Y_u \right\rangle + 2 \left\langle U \Tilde{\beta}^\gamma(u;T), Y_u \right\rangle \right] du \\
        & \qquad \qquad \qquad \qquad \qquad + \mathrm{Tr}(U^\top \Lambda \Lambda^\top) (t - s).
    \end{align*}
    Because $U$ is positive definite, the maximum and minimum of the eigenvalues, $\lambda_{\min}, \lambda_{\max} > 0$, satisfy
    \begin{equation}\label{quad.U.bdd}
               \lambda_{\min} |y|^2 \leq  \langle Uy, y \rangle \leq \lambda_{\max} |y|^2 
    \end{equation}
    for all $y \in \mathbb{R}^m$. Moreover, because the function $\Tilde{\beta}$ is a bounded function of $(t, T)$, there exist positive constants $C$ such that
    \begin{align}\label{beta.til.bdd}
        \frac{2}{\epsilon}| U \Tilde{\beta}(t;T)|^2 + \mathrm{Tr}(U^\top \Lambda \Lambda^\top) \leq C
    \end{align}
    for any $t, T$ with $0 \leq t \leq T$. From the above inequalities (\ref{quad.U.bdd}) and (\ref{beta.til.bdd}), it follows that for $ t \in [0, T - T_1]$, 
    \begin{align*}
        \frac{d}{dt} \mathbb{E}^{\mathbb{Q}^{\gamma}_T}  \left[\langle UY_t, Y_t \rangle \right] &=  \mathbb{E}^{\mathbb{Q}^{\gamma}_T}  \left[ \left\langle  \left(  \Tilde{C}^{\gamma} (t; T)^\top U + U\Tilde{C}^{\gamma} (t; T) \right) Y_t, Y_t \right\rangle + 2 \left\langle U \Tilde{\beta}^\gamma(t;T), Y_t \right\rangle \right] + \mathrm{Tr}(U^\top \Lambda \Lambda^\top) \\
        &\leq - \epsilon \mathbb{E}^{\mathbb{Q}^{\gamma}_T}  \left[|Y_t|^2 \right] + 2 \mathbb{E}^{\mathbb{Q}^{\gamma}_T}  \left[ \langle U \Tilde{\beta}(t;T), Y_t \rangle \right] + \mathrm{Tr}(U^\top \Lambda \Lambda^\top) \\
        &\leq - \frac{\epsilon}{2} \mathbb{E}^{\mathbb{Q}^{\gamma}_T}  \left[|Y_t|^2 \right] + \frac{2}{\epsilon}| U \Tilde{\beta}(t;T)|^2 + \mathrm{Tr}(U^\top \Lambda \Lambda^\top) \\
        & \leq - \frac{\epsilon}{2\lambda_{\max}} \mathbb{E}^{\mathbb{Q}^{\gamma}_T}  \left[\langle UY_t, Y_t \rangle \right] + C,
    \end{align*}
    where the second inequality follows from $\epsilon$-Hölder's inequality.
    By Gronwall's inequality, there exists a constant $C = C(y)$ which depends only on $y$ such that
    \[
        \mathbb{E}^{\mathbb{Q}^{\gamma}_T}  \left[\langle U Y_t, Y_t \rangle \right] \leq C(y), \quad t \in [0, T - T_1].
    \]
    The above inequalities and the right-hand side of (\ref{quad.U.bdd}) lead to 
    \[
        \mathbb{E}^{\mathbb{Q}^{\gamma}_T}  \left[|Y_t|^2 \right] \leq \frac{C(y)}{\lambda_{\min}}, \quad t \in [0, T - T_1].
    \]
    Furthermore, because the length of $[T-T_1, T]$ is $T_1$, there exists a constant $C = C(y, T_1)$ which depends only on $y$ and $T_1$ such that
    \[
        \mathbb{E}^{\mathbb{Q}^{\gamma}_T} \left[ |Y_t|^2 \right] \leq C(y, T_1), \quad t \in [T - T_1, T].
    \]
    As a result, the proposition follows.
\end{proof}

\begin{proposition}\label{LT.bdd}
For $\gamma \in \{q, \alpha + 1, 1\}$,
    \[
        \mathbb{E}^{\mathbb{Q}^{\gamma}_T}  [|L_T|] = O(1), \quad (T \nearrow \infty).
    \]
\end{proposition}
\begin{proof}
    Because $Dr(y) = r_1 + R_2 y, \; D\theta(y) = A,\; \nabla_y Y_t = e^{Bt}$, (\ref{L}) becomes
    \begin{align*}
        L_T &=  \int_0^T ((r_1 + R_2 Y_t)^\top e^{Bt})^\top dt + \int_0^T (A e^{Bt})^\top dW^{\mathbb{Q}^{\gamma}_T}_t \\
        & \quad-\int_0^T (A e^{Bt})^\top \left[ \left\{ (\gamma - 1) A + \Lambda^\top C^\gamma(t;T) \right\} Y_u + (\gamma - 1) a + \Lambda^\top \beta^\gamma(t;T) \right]dt.
    \end{align*} 
    Because $B$ is stable, there exist $M, \omega > 0$ such that 
    \[
        |e^{Bt}| \leq M e^{-\omega t}, \; (t \geq 0).
    \]
    Furthermore, $C^\gamma, \beta^\gamma$ are bounded as functions of $(t, T)$. Therefore, this proposition follows from Proposition~\ref{Y.myopic.upb}.
\end{proof}

\begin{proof}[Proof of Theorem~\ref{excess.turnpike}]
    The result follows from Propositions \ref{excess.hedging.bdd} and \ref{LT.bdd}.
\end{proof}

\begin{proof}[Proof of Theorem~\ref{uni.excess.turnpike}]
    We can prove the theorem in the same way as in the proof of Theorem~\ref{uni.conv.myopic}.
\end{proof}

\subsection{Proofs for Sect.~\ref{Application.sec}}\label{proof.appl.sec}

\begin{lemma}\label{Est.diff.I}
    Let $U_1, U_2$ be utility functions. 
    \begin{itemize}
        \item [(i)] If there exist $f, \Tilde{f}: (0, \infty) \to [0, \infty)$ such that
        \begin{equation}\label{U.prime.growth}
                        U_1^\prime(x + f(x)) \leq U^\prime_2 (x) \leq U_1^\prime(x + \Tilde{f}(x)), \quad x > 0
        \end{equation}
        holds, then 
        \[
            |I_1(z) - I_2(z)| \leq \max\left\{ f(I_2(z)), \Tilde{f}(I_2(z)) \right\}, \quad z > 0
        \]
        holds.
        \item [(ii)] In addition to the assumption in (i), if there exists $g:(0, \infty) \to \mathbb{R}$ such that
        \begin{equation}\label{ART.asp}
                        \left|ART_1(x) - ART_2(x) \right| \leq g(x), \quad x > 0
        \end{equation}
        holds and $U_2$ is the HARA utility, that is, $ART_2(x) = ax + b\; (x > 0)$ for some constants $a, b \in \mathbb{R}$, then 
        \[
            |z I_1^\prime(z) - z I_2^\prime(z)| \leq g(I_1(z)) + |a|\max\left\{ f(I_2(z)), \Tilde{f}(I_2(z)) \right\}, \quad z > 0
        \]
        holds.
    \end{itemize}
\end{lemma}
\begin{proof}
    \begin{itemize}
        \item [(i)] Substituting $x = I_2(z), z > 0$ in (\ref{U.prime.growth}), 
        \[
            U_1^\prime(I_2(z) + f(I_2(z))) \leq z \leq U_1^\prime(I_2(z) + \Tilde{f}(I_2(z)))
        \]
        holds. Because $I_1$ is decreasing, we get 
        \[
            I_2(z) + \Tilde{f}(I_2(z)) \leq I_1(z) \leq I_2(z) + f(I_2(z)),
        \]
        which leads to
        \[
            |I_1(z) - I_2(z)| \leq \max\left\{ f(I_2(z)), \Tilde{f}(I_2(z)) \right\}.
        \]
        \item [(ii)] $z I_i^\prime(z)$ can be represented by \textit{the Arrow--Pratt measure of absolute risk tolerance} $ART_i(x) \coloneqq - \frac{U_i^\prime(x)}{U_i^{\prime\prime}(x)}$ as follows:
        \begin{align*}
            zI_i^\prime(z) = \frac{z}{U_i^{\prime\prime}(I_i(z))} = \frac{U_i^\prime(I_i(z))}{U_i^{\prime\prime}(I_i(z))} = - ART_i(I_i(z)).
        \end{align*}
        Using (\ref{ART.asp}) and the assumption that $U_2$ is the HARA utility, we obtain
        \begin{align*}
            |zI_1^\prime(z) - z I_2^\prime(z)| &= |ART_1(I_1(z)) - ART_2(I_2(z))| \\
                                               &\leq |ART_1(I_1(z)) - ART_2(I_1(z))| +|ART_2(I_1(z)) - ART_2(I_2(z))| \\
                                             &\leq g(I_1(z)) + |a| |I_1(z) - I_2(z)| \\
                                             &\leq g(I_1(z)) + |a|\max\left\{ f(I_2(z)), \Tilde{f}(I_2(z)) \right\},
        \end{align*}                    
        which completes the proof.
    \end{itemize}
\end{proof}
Let 
\begin{align*}
    U_1(x) &\coloneqq \sum_{i = 1}^n \beta_i \frac{(\alpha_i x)^{p_i}}{p_i} = \sum_{i=1}^n w_i \frac{x^{p_i}}{p_i}, \\
    U_2(x) &\coloneqq w_n \frac{x^{p_n}}{p_n},
\end{align*}
where $w_i \coloneqq \beta_i \alpha_i^{p_i} > 0$.
\begin{proof}[Proof of Proposition~\ref{lin.share.diffI}]
First, we prove the inequality (\ref{lin.share.diffI1}). Because $U_1^\prime(x) = \sum_{i=1}^n w_i x^{p_i - 1},\; U_2^\prime(x) = w_n x^{p_n - 1}$, 
    \[
    U_2^\prime(x) \leq U_1^\prime(x), \quad x > 0
    \]
    holds. To look for $\beta \in \mathbb{R}$ that satisfies
    \[
        U_1^\prime(x + x^\beta) \leq U_2^\prime(x)
    \]
 for large enough $x > 0$, we define $h$ by
    \begin{align*}
        h(x) &\coloneqq \frac{U_2^\prime(x) - U_1^\prime(x + x^\beta)}{w_n(x + x^\beta)^{p_{n}-1}} \\
             &= \left( \frac{x}{x + x^\beta} \right)^{p_n - 1} - \sum_{i=1}^{n-1} \frac{w_i}{w_n} (x + x^\beta)^{p_i - p_n} - 1
    \end{align*}
    for all $x > 0$. Then $h^\prime$, a derivative of $h$, is given by 
     \begin{align*}
         h^\prime(x) &= \frac{(p_n - 1)(1 - \beta) x^{p_n - 2 + \beta}}{(x + x^\beta)^{p_n}} - \sum_{i=1}^{n-1} \frac{w_i}{w_n}(p_i - p_n) \frac{(x + x^\beta)^{p_i - 1}}{(x + x^\beta)^{p_n}} (1 + \beta x^{\beta - 1}) \\
         &= - \frac{(p_n - p_{n-1}) x^{p_n + \beta - 2}}{(x + x^\beta)^{p_n}} \left\{ \frac{1 - p_n}{p_n - p_{n-1}}(1 - \beta) - \sum_{i = 1}^{n-1} \frac{w_i}{w_n} \cdot \frac{p_n - p_i}{p_n - p_{n-1}} \cdot (1 + x^{\beta - 1})^{p_i - 1} \left( \beta x^{p_i - p_n} + x^{1 + p_i - p_n - \beta} \right) \right\}.
     \end{align*}
    If $\beta \in (1 + p_{n-1} -p_n, 1 )$, then
    \[
        1 + p_i - p_n - \beta < 1 + p_{n-1} - p_n - \beta < 0, \quad(i = 1, \dots, n-1),
    \]
    which means that
    \[
     \sum_{i = 1}^{n-1} \frac{w_i}{w_n} \cdot \frac{p_n - p_i}{p_n - p_{n-1}} \cdot (1 + x^{\beta - 1})^{p_i - 1} \left( \beta x^{p_i - p_n} + x^{1 + p_i - p_n - \beta} \right) \rightarrow \sum_{i = 1}^{n-1} \frac{w_i}{w_n} \cdot \frac{p_n - p_i}{p_n - p_{n-1}} \cdot (1 + 0)^{p_i - 1} \left( \beta \cdot 0 + 0 \right) = 0
    \]
    as $x \nearrow \infty$.
    Therefore, $h^\prime(x) < 0$ holds for large enough $x > 0$. Combining this with the fact that
    \begin{align*}
        h(x) &= \left( \frac{1}{1 + x^{\beta - 1}} \right)^{p_n - 1} - \sum_{i=1}^{n-1} \frac{w_i}{w_n} x^{p_i - p_n} (1 + x^{\beta - 1})^{p_i - p_n} - 1 \\
        &\rightarrow 1 - \sum_{i=1}^{n-1} \frac{w_i}{w_n} \cdot 0 \cdot (1 + 0)^{p_i - p_n} - 1 \quad (x \nearrow \infty) \\
        &= 0
    \end{align*}
    leads to $h(x) > 0$ for large enough $x > 0$. Using Lemma~\ref{Est.diff.I}(i) with $f(x) = x^\beta, \Tilde{f}(x) = 0,$ and $\beta \in (1 + p_{n-1} -p_n, 1 )$ leads to 
    \[
        |I_1(z) - I_2(z)| \leq \left( \frac{z}{w_n} \right)^{\beta(q_n - 1)}
    \]
    for small enough $z > 0$, which means that the inequality (\ref{lin.share.diffI1}) holds for any nonnegative $\beta \in (1 + p_{n-1} - p_n, 1)$ and a constant $K > 0$. 
    Next, we prove the inequality (\ref{lin.share.diffI2}). Because
    \begin{align*}
       \left| ART_1(x) - ART_2(x) \right| &= \left| \frac{\sum_{i=1}^n w_i x^{p_i - 1}}{\sum_{i=1}^n w_i (1 - p_i) x^{p_i - 2}} -  \frac{x}{1 - p_n} \right| \\
                            &= \frac{\sum_{i=1}^{n-1} (p_n - p_i) w_i x^{p_i - 1} }{(1 - p_n) \left( \sum_{i = 1}^{n} w_i (1 - p_i) x^{p_i - 2} \right)}  \\
                            &\leq \frac{\sum_{i=1}^{n-1} (p_n - p_i) w_i x^{p_i - 1} }{ (1 - p_n)^2 w_n x^{p_n - 2} } \\
                            &= \sum_{i=1}^{n-1} \frac{p_n - p_i}{(1 - p_n)^2} \cdot \frac{w_i}{w_n} \cdot x^{p_i - p_n + 1} \\
                            &\leq Cx^{p_{n-1} - p_n + 1}
    \end{align*}
    holds for large enough $x > 0$ and a constant $C > 0$, Lemma~\ref{Est.diff.I}(ii) implies that for small enough $z > 0$,
    \begin{align*}
        |zI_1^\prime(z) - z I_2^\prime(z)| &\leq C I_1(z)^{p_{n-1} - p_n + 1} + \frac{1}{1 - p_n} \left( \frac{z}{w_n} \right)^{\beta(q_n - 1)} \\
        &\leq C \left( \left( \frac{z}{w_n} \right)^{q_n - 1} + \left( \frac{z}{w_n} \right)^{\beta (q_n - 1)} \right)^{(p_{n-1} - p_n + 1)^+} + \frac{1}{1 - p_n} \left( \frac{z}{w_n} \right)^{\beta(q_n - 1)} \\
        &\leq \Tilde{C} \left( z^{(q_n - 1) (p_{n-1} - p_n + 1)^+} + z^{\beta (q_n - 1)(p_{n-1} - p_n + 1)^+} + z^{\beta (q_n - 1)} \right) \\
        &\leq 3 \Tilde{C} z^{\beta (q_n - 1)},
    \end{align*}
    where $\Tilde{C}$ is a constant, and the last inequality follows from 
    \[
        \beta(q_n - 1) \leq (q_n - 1) (p_{n-1} - p_n + 1)^+ \leq \beta (q_n - 1) (p_{n-1} - p_n + 1)^+ \leq 0.
    \]
    We have completed the proof.
\end{proof}

\appendix
\section{Appendix: Malliavin calculus}
We recall some results on Malliavin calculus which allow us to derive an explicit stochastic flow representation for the optimal portfolio process $\hat{\pi}$. For details, see \cite{N, OK, PS, SH}.

Consider a complete probability space $(\Omega, \mathcal{F}, P)$ and a standard $n$-dimensional Brownian motion $W = (W^1, \dots, W^n)^\top$ defined on $(\Omega, \mathcal{F}, P)$. We denote by $(\mathcal{F}_t)_{t \geq 0}$ the $P$-augmentation of the natural filtration generated by $W = (W_t)_{t \geq 0}$. 

We introduce the Malliavin derivative operator as in \cite{OK}. Fix $T > 0$.
Let $C_b^\infty(\mathbb{R}^m)$ be the space of infinitely differentiable functions on $\mathbb{R}^m$ which, together with all partial derivatives, are bounded. By $\mathscr{S}$ we denote the class of {\it smooth random variables}, namely, random variables of the form 
\[
    F = f(W_{t_1}, \dots, W_{t_m}),
\]
where $(t_1, \dots, t_m) \in [0, T]^m$ and the function $f = f(x^{11}, \dots, x^{n1}, \dots, x^{1m}, \dots, x^{nm})$ belongs to $C_b^\infty(\mathbb{R}^{nm})$.
For each $F \in \mathscr{S}$, the {\it Malliavin derivative} of $F$ is the $L^2([0, T];\mathbb{R})^n$-valued random variable $DF = (D^1F, \dots D^nF)$ with components
\[
    D^iF (\cdot) \coloneqq \sum_{j = 1}^m \frac{\partial f}{\partial x^{ij}}(W_{t_1}, \dots, W_{t_m}) 1_{[0, t_j]}(\cdot), \quad (i = 1, \dots, n).
\]
Fix $p \in [1, \infty)$. Because we can view the operator $D$ as an operator from $L^p(\Omega;\mathbb{R})$ to $L^p \left( \Omega; (L^2[0, T];\mathbb{R})^n \right)$ and $D$ is closable, we denote the closure of $D$ again by $D$ and the domain of $D$ in $L^p(\Omega;\mathbb{R})$ by $\mathbb{D}_{p, 1}$. Thus, $\mathbb{D}_{p, 1}$ is the closure of $\mathscr{S}$ with respect to the norm
\begin{align*}
    ||F||_{p, 1} &\coloneqq ||F||_{L^p(\Omega; \mathbb{R})} + ||DF||_{L^p \left( \Omega; (L^2[0, T];\mathbb{R})^n \right)} \\
    &= \mathbb{E}[|F|^p]^{\frac{1}{p}} + \mathbb{E} \left[ \left( \int_0^T |DF(t)|^2 dt \right)^{\frac{p}{2}} \right]^\frac{1}{p}.
\end{align*}
Here, $|\cdot|$ denotes the Euclidean norm on $\mathbb{R}^n$. 
$\mathbb{D}_{p,1}$ is a Banach space with respect to the norm $||\cdot ||_{p, 1}$.
Given $F \in \mathbb{D}_{p, 1}$, $DF$ is an $(L^2[0, T])^n$-valued random variable. To each $DF$, we can find a measurable process $[0, T] \times \Omega \ni (t, \omega) \mapsto D_tF(\omega) \in \mathbb{R}^n$ such that for almost all $\omega \in \Omega$, $D_tF (\omega) = DF(\omega)(t)$ holds for almost everywhere $t \in [0, T]$. Therefore, we identify $(L^2[0, T])^n$-valued random variable $DF$ with $\mathbb{R}^n$-valued measurable process $ (t, \omega) \mapsto D_t F(\omega)$ without further comment. 
\begin{remark}
    Note that for real-valued random variable $F \in \mathbb{D}_{p, 1}$, we define $DF = (D^1F, \dots, D^nF)$ as a row vector; that is, $DF$ is an $\mathbb{R}^{1 \times n}$-valued stochastic process. For $\mathbb{R}^m$-valued random variable $F = (F_1, \dots, F_m)^\top \in \mathbb{D}_{p, 1}^m$ we define $DF = (D^jF_i)_{\substack{1 \leq i \leq m \\ 1 \leq j \leq n}}$, which is an $\mathbb{R}^{m \times n}$-valued stochastic process.
\end{remark}
We collect well-known results on Malliavin calculus, that is, Clark's formula (Proposition~\ref{Clark.formula}), chain rule (Proposition~\ref{chain.rule}), and Malliavin derivatives of Lebesgue integrals, stochastic integrals, and solutions to SDEs (Propositions \ref{diff.under.leb}--\ref{diff.sol.SDE}).
Firstly, we quote Clark's formula for random variables in $\mathbb{D}_{1, 1}$, which comes from \cite{KOL}.
\begin{proposition}\label{Clark.formula}
    For every $F \in \mathbb{D}_{1, 1}$ we have
    \[
        \mathbb{E}[F | \mathcal{F}_t] = \mathbb{E}[F] + \int_0^t \mathbb{E} \left[ D_sF | \mathcal{F}_s\right] dW_s, \quad t \in [0, T].
    \]
\end{proposition}
The following proposition is a straightforward multidimensional version of Lemma~A.1 in \cite{OK}.
\begin{proposition}\label{chain.rule}
    Let $F = (F_1, \dots F_m)^\top \in \mathbb{D}_{1,1}^m$. Let $\phi = (\phi^1, \dots, \phi^k)^\top \in C^1(\mathbb{R}^m ;\mathbb{R}^k)$. Assume that
    \begin{equation} \label{asp.chain.mal}
        \mathbb{E}[|\phi^l(F)|] + \mathbb{E} \left[ \left( \int_0^T \left| \sum_{i = 1}^m \frac{\partial \phi^l}{\partial x_i} (F) D_tF_i \right|^2 dt \right)^{\frac{1}{2}} \right] < \infty    
    \end{equation}
    for all $l = 1, \dots k$. Then $\phi(F) \in \mathbb{D}_{1, 1}^k$ and 
    \[
        D(\phi(F)) = D\phi(F) DF,
    \]
    where $D\phi = \left( \frac{\partial\phi^i}{\partial x_j} \right)_{\substack{1 \leq i \leq k \\ 1 \leq j \leq m}}:\mathbb{R}^m \to \mathbb{R}^{k \times m}$ is the Jacobi matrix of $\phi$.
\end{proposition}
\begin{remark}
    Hölder's inequality implies that the condition (\ref{asp.chain.mal}) holds if $F \in \mathbb{D}_{p, 1}^m,\; \phi^l(F) \in L^1,\; \frac{\partial \phi^l}{\partial x_i}(F) \in L^q,\; l = 1, \dots k,\; i = 1, \dots m$ for some $p \in [1, \infty),\; q \in (1, \infty]$ such that $\frac{1}{p} + \frac{1}{q} = 1$. In particular, the condition (\ref{asp.chain.mal}) holds if $F \in \bigcap_{p \geq 1} \mathbb{D}_{p, 1}^m$ and $\phi, D\phi$ are of polynomial growth.  
\end{remark}
Proposition~\ref{diff.under.leb} is Lemma~5.1 in \cite{La}, Proposition~\ref{diff.under.SI} is Proposition~2.3 in \cite{OK}, and Proposition~\ref{diff.sol.SDE} is Proposition~8.2 in \cite{SH}.
\begin{proposition}\label{diff.under.leb}
    Let $u = (u_s)_{s \in [0, T]}$ be a real-valued, continuous, progressively measurable process such that
    \begin{enumerate}
        \item [(i)] $u_s \in \mathbb{D}_{1, 1}$ for every $s \in [0, T]$,
        \item [(ii)] $\sup_{s \in [0, T]} \mathbb{E}\left[|u_s|^q \right] < \infty$ for some $q > 1$, and $ \sup_{s \in [0, T]} \mathbb{E} \left[ \int_0^T |D^j_tu_s|^4 dt  \right] < \infty$ for $j = 1, \dots, n$,
        \item [(iii)] $s \mapsto D_tu(s, \omega)$ is left (or right) continuous for almost every $(t, \omega) \in [0, T] \times \Omega$.
    \end{enumerate}
    Then $\int_0^Tu_s ds \in \mathbb{D}_{1, 1}$ and $D_t \int_0^T u_s ds = \int_t^T D_tu_s ds$.
\end{proposition}

\begin{proposition}\label{diff.under.SI}
    Let $u = (u^1, \dots, u^n)^\top$ be an $\mathbb{R}^n$-valued progressively measurable process such that
    \begin{itemize}
        \item [(i)] $u_s \in \mathbb{D}_{1, 1}^n$ for every $s \in [0, T]$,
        \item [(ii)] $[0, T] \times \Omega \ni (x, \omega) \mapsto Du(s, \omega) \in (L^2[0, T])^{n^2}$ admits a progressively measurable version,
        \item [(iii)] \begin{align*}
                    |||u|||_{1, 1} &\coloneqq \mathbb{E} \left[ \left( \int_0^T |u_s|^2 ds \right)^{\frac{1}{2}} \right] + \mathbb{E} \left[ \left\{ \int_0^T \int_0^T \left| D_t u_s \right|^2 dt ds \right\}^{\frac{1}{2}} \right] < \infty,
        \end{align*}
        where $| \cdot|$ denotes Euclidean norm on $\mathbb{R}^{n \times n}$.
    \end{itemize}
    Then $\int_0^T u^\top_s dW_s \in \mathbb{D}_{1, 1}$ and 
    \[
        D_t \int_0^T u^\top_s dW_s = u_t^\top + \left( \int_0^T (D_tu_s)^\top dW_s \right)^\top .
    \]
\end{proposition}

\begin{proposition}\label{diff.sol.SDE}
    For $d \in \mathbb{N}$, we consider the $d$-dimensional SDE
    \begin{equation} \label{Append.SDE}
        dX_t = \mu(X_t) dt + \sigma(X_t) dW_t, \quad X_0 = x \in \mathbb{R}^d,
    \end{equation}
    where $\mu = (\mu_1, \dots \mu_d)^\top :\mathbb{R}^d \to \mathbb{R}^d, \; \sigma = (\sigma_{i,j})_{\substack{1 \leq i \leq d \\ 1 \leq j \leq n}}:\mathbb{R}^d \to \mathbb{R}^{d \times n}$ are continuously differentiable and satisfy
    \[
        \sup_{x \in \mathbb{R}^d} \left( \left| \frac{\partial\mu_i}{\partial x_k}(x) \right| + \left| \frac{\partial \sigma_{i, j}}{\partial x_k} (x) \right| \right) < \infty
    \]
    for $i, k = 1, \dots, d, \; j = 1, \dots, n.$
    Then (\ref{Append.SDE}) has a unique strong solution $X = (X^1, \dots, X^d)^\top$ which satisfies the following:
    \begin{itemize}
        \item [(i)] $X^k_s \in \bigcap_{p \geq 1} \mathbb{D}_{p, 1}, \quad k = 1, \dots d, \; s \in [0, T]$;
        \item [(ii)] $D_t X_s$ satisfies
    \[
        D_t X_s = \sigma(X_t) + \int_t^s D\mu(X_u) D_t X_u du + \sum_{j = 1}^n \int_t^s D \sigma_{\cdot j}(X_u) D_t X_u dW^j_u
    \]
    for $t \in [0, s]$ and $D_t X_s = 0$ for $t \in (s, T]$, where $D\mu = \left( \frac{\partial \mu_i}{\partial x_j} \right)_{\substack{1 \leq i \leq d \\ 1 \leq j \leq d}}, \; D\sigma_{\cdot j} = \left( \frac{\sigma_{i,j}}{\partial x_l} \right)_{\substack{ 1 \leq i \leq d \\ 1 \leq l \leq d}}$;
    \item [(iii)] for $j = 1, \dots, n, \; p \in [1, \infty),$
        \[
            \sup_{r \in [0, T] } \mathbb{E} \left[ \sup_{s \in [0, T]} \left| D^j_r X^k_s \right|^p \right] < \infty ;
        \]

    \item [(iv)] $D_t X_s = \nabla_x X_s (\nabla_x X_t)^{-1} \sigma(X_t)$ for $t \in [0, s]$, where $\nabla_x X$ is an $\mathbb{R}^{d \times d}$-valued stochastic process satisfying
    \[
        \nabla_x X_s = I + \int_0^s D\mu(X_u) \nabla_x X_u du + \sum_{j = 1}^n \int_t^s D \sigma_{\cdot j}(X_u) \nabla_x X_u dW^j_u
    \]
    for $s \in [0, T]$ and $I \in \mathbb{R}^{d \times d}$ is the identity matrix. 
    \end{itemize}
\end{proposition}

\section{Appendix: Option pricing theory with stochastic factor models in complete markets}\label{opt.pricing}
In this appendix, we recall well-known results on pricing and hedging problems in a complete market with a stochastic factor process (\ref{gen.stoch.fac.model}).
Under the equivalent martingale measure $\mathbb{Q}$, the dynamics of risky assets $S$ and the stochastic factor process $Y$ are denoted by
\begin{align*}
    dS_t &= \diag(S_t) \left( \mathbf{1}r(Y_t) dt + \sigma(Y_t) dW^{\mathbb{Q}}_t \right),&  S_0 &= s_0 \in \mathbb{R}^n_{++},\\
    dY_t &= \Tilde{b}(Y_t) dt + a(Y_t) dW^{\mathbb{Q}}_t,&  Y_0 &= y \in \mathbb{R}^m,
\end{align*}
where $\Tilde{b}(y) \coloneqq b(y) - a(y) \theta(y)$.
Let $\mathcal{L}$ be a generator of $(S, Y)$ under $\mathbb{Q}$; that is, for $f:[0, T] \times \mathbb{R}^n_{++} \times \mathbb{R}^m \to \mathbb{R}$, $\mathcal{L}f:[0, T] \times \mathbb{R}^n_{++} \times \mathbb{R}^m \to \mathbb{R}$ is defined by
    \[
        \mathcal{L}f \coloneqq D_s f^\top r(y) s + D_y f^\top \Tilde{b}(y) + \frac{1}{2} \mathrm{Tr} \left[ \left\{ \Sigma \Sigma^\top \right\} (s, y) D^2 f \right], \quad (t, s, y) \in [0, T] \times \mathbb{R}^n_{++} \times \mathbb{R}^m, 
    \]
where $\Sigma(s, y) \coloneqq \begin{pmatrix}
    \diag(s) \sigma(y) \\ a(y)
\end{pmatrix} \in \mathbb{R}^{(n+m) \times n}$.

\begin{theorem}\label{comp.stoc.hedg}
Let $\Phi : \mathbb{R}^n_{++} \to \mathbb{R}$ and $u :[0, T] \times \mathbb{R}^n_{++} \times \mathbb{R}^m \to \mathbb{R}$ be a solution of the Cauchy problem
\begin{align*}
    \partial_t u + \mathcal{L}u - r(y) u &= 0,&  &on \quad \left[0, T\right) \times \mathbb{R}^n_{++} \times \mathbb{R}^m, \\
    u(T, s, y) &= \Phi(s),& &on \quad \mathbb{R}^n_{++} \times \mathbb{R}^m.
\end{align*}
\begin{enumerate}
    \item [(i)] Let $x \coloneqq u(0, s_0, y)$ and $\pi = (\pi_t)_{t \in [0, T]}$ be a portfolio process satisfying
    \[
        \pi_t^\top \sigma(Y_t) = D_s u^\top \diag (S_t) \sigma(Y_t) + D_y u^\top a(Y_t),
    \]
    where $D_su$ and $D_y u$ are evaluated at $(t, S_t, Y_t)$. Then $x$ is the replicating cost and $\pi$ is the hedging portfolio. Indeed, 
    \[
        X^{x, \pi}_t = u(t, S_t, Y_t), \qquad t \in [0, T].
    \]
    In particular,
    \[
        X^{x, \pi}_T = \Phi(S_T).
    \]
    \item [(ii)] Moreover, if $(H_t X^{x, \pi}_t)_{t \in [0, T]}$ is a $\mathbb{P}$-martingale, then
    \begin{equation}\label{pricing.formula}
        u(t, s, y) = \mathbb{E}^{\mathbb{Q}}\left[ \left. \exp \left( - \int_t^T r(Y_v) dv \right) \Phi \left( S_T \right) \right| S_t = s, Y_t = y \right].
    \end{equation}
    In particular, the replicating cost $x$ is given by
    \begin{equation}\label{replicating.cost}
        x \coloneqq u(0, s_0, y) = \mathbb{E}^{\mathbb{Q}}\left[\exp \left( - \int_0^T r(Y_v) dv \right) \Phi \left( S_T \right) \right]. 
    \end{equation}
\end{enumerate} 
\end{theorem}
\begin{proof}
    By the Ito formula, 
    \begin{align*}
        \frac{1}{S^0_t} u(t, S_t, Y_t) &= u(0, s_0, y) + \int_0^{t} \frac{1}{S^0_v} \left\{ \left( \partial_t u + \mathcal{L}u - r u \right)(v, S_v, Y_v) dv + \left( D_su^\top \diag (S_v) \sigma(Y_v) + D_y u^\top a(Y_v) \right) dW^{\mathbb{Q}}_v \right\} \\
        &=  u(0, s_0, y) + \int_0^{t} \frac{1}{S^0_v} \left( D_su^\top \diag (S_v) \sigma(Y_v) + D_y u^\top a(Y_v) \right) dW^{\mathbb{Q}}_v \\
        &=  x + \int_0^{t} \frac{1}{S^0_v} \pi_v^\top \sigma(Y_v)  dW^{\mathbb{Q}}_v \\
        &= \frac{X^{x, \pi}_t}{S^0_t},
    \end{align*}
    which proves (i). If $(H_t X^{\pi}_t)_{t \in [0, T]}$ is a $\mathbb{P}$-martingale, then
    \begin{align*}
        u(t, S_t, Y_t) &= X^{x, \pi}_t \\
                       &= \frac{1}{H_t} \mathbb{E}_t \left[ H_T X^{x, \pi}_T \right] \\
                       &= \mathbb{E}^{\mathbb{Q}}_t\left[  \exp \left( - \int_t^T r(Y_v) dv \right) \Phi \left( S_T \right) \right].
    \end{align*}
    By the Markov property of $(S, Y)$, we obtain (\ref{pricing.formula}) and (\ref{replicating.cost}).
\end{proof}

\section{Appendix: Relationship between stochastic control methods and martingale methods}
In this appendix, we verify that the terminal wealth obtained via the dynamic programming approach matches the optimal terminal wealth derived from martingale duality methods. We prove the result in complete stochastic factor models given by (\ref{gen.stoch.fac.model}).
    \begin{theorem}\label{control.equiv.dual}
        Let $V:[0, T] \times (0, \infty) \times \mathbb{R}^m$ be a classical solution to the HJB equation. Let $\hat{\pi}:[0, T] \times (0, \infty) \times \mathbb{R}^m$ be a candidate for optimal investment strategies. Furthermore, we assume that 
        \[
            \lim_{t \nearrow T} V_x (t, x, y) = U^\prime(x), \quad (x, y) \in (0, \infty) \times \mathbb{R}^m
        \] and $(H_t X^{x, \hat{\pi}}_t )_{t \in [0, T]}$ is a martingale for some $x > 0$.
        Then
        \[
            X^{x, \hat{\pi}}_T = I(\hat{\lambda} H_T),
        \]
        and $\hat{\pi}$ is the optimal feedback strategy. 
    \end{theorem}
    \begin{proof}
        We rewrite the HJB equation as 
        \begin{equation}\label{HJB.pi}
            V_t + r(y) x V_x + b(y)^\top D_y V + \frac{1}{2} \mathrm{Tr} \left( a a^\top (y) D^2_{yy}V \right) - \frac{V_{xx}}{2} | \sigma^\top (y) \hat{\pi}(t, x, y)|^2 = 0.
        \end{equation}

        Differentiating (\ref{HJB.pi}) with respect to $x$, we have
        \begin{equation}
            \begin{multlined}
                  V_{tx} + r(y)  V_x + r(y) x V_{xx} + b(y)^\top D_y V_x + \frac{1}{2} \mathrm{Tr} \left( a a^\top (y) D^2_{yy}V_x \right) \\ + \hat{\pi}^\top \sigma(\theta V_{xx}  + a^\top D_y V_{xx}) + \frac{V_{xxx}}{2} | \sigma^\top (y) \hat{\pi}(t, x, y)|^2 = 0.
            \end{multlined}
        \end{equation}
        That is, 
        \[
            V_{tx} + r(y) V_x + \mathcal{L}^{\hat{\pi}} V_x = 0,
        \]
        where $\mathcal{L}^{\pi}$ is a controlled generator: 
        \[
            \mathcal{L}^{\pi} f \coloneqq r(y) x f_x + b(y)^\top D_y f + \frac{1}{2} \mathrm{Tr} \left( a a^\top(y) D_{yy}^2 f \right) + \pi^\top \sigma(y) \left( \theta(y) f_x + a^\top(y) D_y f_x\right) + \frac{1}{2} | \sigma^\top(y) \pi|^2 f_{xx} = 0.
        \]
        By the Ito formula, 
        \begin{align*}
            dV_x(t, X^{x, \hat{\pi}}_t, Y_t) &= (V_{xt} + \mathcal{L}^{\hat{\pi}_t} V_x) dt + \left\{ V_{xx} \hat{\pi}^\top_t \sigma(Y_t) + D_y V_x^\top a (Y_t) \right\} dW_t \\
            &= V_x \left( -r(Y_t) dt - \theta(Y_t)^\top dW_t \right).
        \end{align*}
        Therefore,
        \[
            V_x(t, X^{x, \hat{\pi}}_t, Y_t) = V_x (0, x, y)H_t.
        \]
        Let $t \nearrow T$, then
        \[
            U^\prime (X^{x, \hat{\pi}}_T) = V_x (0, x, y) H_T,
        \]
        which leads to 
        \[
            X^{x, \hat{\pi}}_T = I \left( V_x (0, x, y) H_T \right).
        \]
        Because $H_t X^{x, \hat{\pi}}_t$ is a martingale, 
        \[
            x = \mathbb{E}\left[ H_T X^{x, \hat{\pi}}_T \right] = \mathbb{E}\left[ H_T I \left( V_x (0, x, y) H_T \right) \right]
        \]
        holds and $\hat{\lambda} =  V_x (0, x, y)$ by the uniqueness of $\hat{\lambda}$. We have completed the proof.
    \end{proof}

\section{Appendix: Matrix Riccati equation}
We recall some facts stated in \cite{Ku} about matrix Riccati differential equations.
Let $T > 0,\; A \in \mathbb{R}^{n \times n},\; B \in \mathbb{R}^{n \times m},\; C \in \mathbb{R}^{m \times n}$.
We consider an $n \times n$ matrix solution $P = P( \cdot \;; T):[0, T] \to \mathbb{R}^{n \times n}$ of the Riccati differential equation
\begin{equation}\label{MRDE}
    \begin{aligned}
        &\dot{P}(t) - P(t) B B^\top P(t) + A^\top P(t) + P(t)A + C^\top C = 0, \quad t \in [0, T], \\
        &P(T) =0.
    \end{aligned}
\end{equation}
First, we state the existence and uniqueness of (\ref{MRDE}).
\begin{theorem}\label{exist.uniq.RDE}
\begin{enumerate}
    \item [(i)] There exists a nonnegative unique solution $P:[0, T] \to \mathbb{S}^n_+ $ to (\ref{MRDE}).
    \item [(ii)] For any $0 \leq t_1 \leq t_2 \leq T$,
    \[
        0 \leq P(t_2) \leq P(t_1)
    \]
    holds.
\end{enumerate}
\end{theorem}
Next, we state the asymptotic behaviors of the solution.
\begin{definition}
\begin{enumerate}
    \item [(i)] The pair $(A, B)$ is said to be stabilizable if a matrix $L \in \mathbb{R}^{m \times n}$ exists such that $A + BL$ is stable (i.e., all its eigenvalues have negative real parts).
    \item [(ii)] The pair $(C, A)$ is said to be detectable if a matrix $F \in \mathbb{R}^{n \times m}$ exists such that $FC + A$ is stable.
\end{enumerate}
\end{definition}

\begin{theorem}\label{Asmp.RDE}
    \begin{enumerate}
        \item [(i)] If $(A, B)$ is stabilizable, then there exists a finite nonnegative matrix $P_\infty \in \mathbb{S}^n_+$ such that
        \[
            \lim_{T \nearrow \infty} P(t;T) = P_\infty 
        \] 
        for any $t$, and $P_\infty$ satisfies the algebraic Riccati equation
        \[
             - P_\infty B B^\top P_\infty + A^\top P_\infty + P_\infty A + C^\top C = 0.
        \]
        \item [(ii)] If $(A, B)$ is stabilizable and $(C, A)$ is detectable, then $A - B B^\top P_\infty$ is stable.
    \end{enumerate}
\end{theorem}
The following theorem is Lemma~4.4 in \cite{KN}.
\begin{theorem}\label{Asmp.LODE}
    Let $K:[0, \infty) \to \mathbb{R}^{n \times n}, \; g:[0, \infty) \to \mathbb{R}^n$, and $f:[0, \infty) \to \mathbb{R}^n$ be a solution of the ordinary differential equation
    \[
        \dot{f}(t) = K(t) f(t) + g(t), \quad f(0) = 0.
    \]
    Suppose that $K(t) \rightarrow \Tilde{K} \in \mathbb{R}^{n \times n}$ and $ g(t) \rightarrow \Tilde{g} \in \mathbb{R}^n$ as $t \nearrow \infty$ and $\Tilde{K}$ is a stable matrix. Then there exists $\Tilde{f} = \lim_{t \nearrow \infty}f(t)$ that satisfies the equation 
    \[
        \Tilde{K} \Tilde{f} + \Tilde{g} = 0.
    \]
\end{theorem}

\subsection*{Acknowledgement}
This work was supported by JST SPRING, Grant Number JPMJSP2138.

\end{document}